\documentclass{article}
\usepackage{algorithm, algpseudocode}
\usepackage{amsmath}
\usepackage{amsthm}
\usepackage{comment}
\usepackage{multirow}
\usepackage{authblk}
\usepackage{txfonts}
\usepackage[top=30truemm,bottom=30truemm,left=25truemm,right=25truemm]{geometry}

\newcommand*{\DtR}[1]{{DtR(#1)}}
\newcommand*{\DtQ}[2]{{DtQ(#1, #2)}}
\newcommand*{\C}[2]{{C^{#1}_{#2}}}
\newcommand*{\Ch}[2]{{\hat{C}^{#1}_{#2}}}
\newcommand*{\Chat}[1]{{\hat{C}_{#1}}}
\newcommand*{\Cu}[1]{{C_{u_{#1}}}}
\newcommand*{\Cuu}[2]{{C^{#1}_{u_{#2}}}}

\newtheorem{theorem}{$Theorem$}
\newtheorem{definition}{$Definition$}
\newtheorem{lemma}{$Lemma$}
\newtheorem{corollary}{$Corollary$}

\theoremstyle{remark}
\newtheorem*{remark}{Remark}

\title{Uniform Partition in Population Protocol Model under Weak Fairness}

\author{Hiroto Yasumi}
\author{Fukuhito Ooshita}
\author{Michiko Inoue}
\date{}
\affil{Nara Institute of Science and Technology}

\begin{document}
\maketitle

\begin{abstract}
We focus on a uniform partition problem in a population protocol model.
The uniform partition problem aims to divide a population into $k$ groups of the same size, where $k$ is a given positive integer.
In the case of $k=2$ (called uniform bipartition), a previous work clarified space complexity under various assumptions: 1) an initialized base station (BS) or no BS, 2) weak or global fairness, 3) designated or arbitrary initial states of agents, and 4) symmetric or asymmetric protocols, except for the setting that agents execute a protocol from arbitrary initial states under weak fairness in the model with an initialized base station. In this paper, we clarify the space complexity for this remaining setting. In this setting, we prove that $P$ states are necessary and sufficient to realize asymmetric protocols, and that $P+1$ states are necessary and sufficient to realize symmetric protocols, where $P$ is the known upper bound of the number of agents.
From these results and the previous work, we have clarified the solvability of the uniform bipartition for each combination of assumptions.
Additionally, we newly consider an assumption on a model of a non-initialized BS and clarify solvability and space complexity in the assumption.
Moreover, the results in this paper can be applied to the case that $k$ is an arbitrary integer (called uniform $k$-partition).

\end{abstract}

\section{Introduction}
\subsection{The Background}
\begin{table}[t!]
\begin{minipage}{1.0\hsize}
\caption{The minimum number of states to solve the uniform bipartition problem under global fairness.}
\begin{center}
\label{tab:global}
\begin{tabular}{|c|c|c|c|c|c|}
\hline
BS & initial states of agents & symmetry & upper bound & lower bound & paper \\
\hline
\multirow{4}{*}{initialized BS} & \multirow{2}{*}{designated} & asymmetric & 3 & 3 & \multirow{6}{*}{\cite{bipartition}} \\
\cline{3-5}
 &  & symmetric & 3 & 3 &  \\
\cline{2-5} 
 & \multirow{2}{*}{arbitrary} & asymmetric & 4 & 4 & \\
\cline{3-5}
 &  & symmetric & 4 & 4 & \\
\cline{1-5}
\multirow{4}{*}{non-initialized BS} & \multirow{2}{*}{designated} & asymmetric & 3 & 3 & \\
\cline{3-5}
 &  & symmetric & 3 & 3 & \\
\cline{2-6} 
 & \multirow{2}{*}{arbitrary} & asymmetric & \multicolumn{2}{|c|}{unsolvable} & this \\
\cline{3-5}
 &  & symmetric & \multicolumn{2}{|c|}{unsolvable} & paper \\
\hline
\multirow{4}{*}{no BS} & \multirow{2}{*}{designated} & asymmetric & 3 & 3 & \multirow{4}{*}{\cite{bipartition}} \\
\cline{3-5}
 &  & symmetric & 4 & 4 &  \\
\cline{2-5}
 & \multirow{2}{*}{arbitrary} & asymmetric & \multicolumn{2}{|c|}{unsolvable} &  \\
\cline{3-5}
 &  & symmetric & \multicolumn{2}{|c|}{unsolvable} &  \\
\hline
\end{tabular}
\end{center}
\caption{The minimum number of states to solve the uniform bipartition problem under weak fairness.}
\begin{center}
\label{tab:weak}
\begin{tabular}{|c|c|c|c|c|c|}
\hline
BS & initial states of agents & symmetry & upper bound & lower bound & paper \\
\hline
\multirow{4}{*}{initialized BS} & \multirow{2}{*}{designated} & asymmetric & 3 & 3 & \multirow{2}{*}{\cite{bipartition}} \\
\cline{3-5}
 &  & symmetric & 3 & 3 & \\
\cline{2-6} 
 & \multirow{2}{*}{arbitrary} & asymmetric & $P$ & $P$ & this\\
\cline{3-5}
 &  & symmetric & $P+1$ & $P+1$ & paper \\
\hline
\multirow{4}{*}{non-initialized BS} & \multirow{2}{*}{designated} & asymmetric & 3 & 3 & \multirow{2}{*}{\cite{bipartition}} \\
\cline{3-5}
 &  & symmetric & 3 & 3 & \\
\cline{2-6} 
 & \multirow{2}{*}{arbitrary} & asymmetric & \multicolumn{2}{|c|}{unsolvable} & this \\
\cline{3-5}
 &  & symmetric & \multicolumn{2}{|c|}{unsolvable} & paper \\
\hline
\multirow{4}{*}{no BS} & \multirow{2}{*}{designated} & asymmetric & 3 & 3 & \multirow{4}{*}{\cite{bipartition}} \\
\cline{3-5}
 &  & symmetric & \multicolumn{2}{|c|}{unsolvable} & \\
\cline{2-5}
 & \multirow{2}{*}{arbitrary} & asymmetric & \multicolumn{2}{|c|}{unsolvable} & \\
\cline{3-5}
 &  & symmetric & \multicolumn{2}{|c|}{unsolvable} & \\
\hline
\end{tabular}
\end{center}
\end{minipage}
\end{table}

A population protocol model \cite{angluin2006computation,angluin2005self} is an abstract model for devices with heavily limited computation and communication capability. 
The devices are represented as passively moving agents, and a set of agents is called a population.
In this model, if two agents approach, an interaction happens between them. At the time of the interaction, the two agents update their states.
By repeating such interactions, agents proceed with computation.
The population protocol model has many application examples such as sensor networks and molecular robot networks.
For example, one may construct a network to investigate the ecosystem by attaching sensors to a flock of wild small animals such as birds. In this system, sensors exchange information with each other when two sensors approach sufficiently close. By repeating such information exchange, the system eventually grasps the entire environment of the flock.
Another example is a system of molecular robots \cite{murata2013molecular}. In this system, a large number of robots cooperate in a human body to achieve an objective (e.g. carrying medicine).
To realize such systems, various fundamental protocols have been proposed in the population protocol model \cite{aspnes2009introduction}.
	For example, there are leader election protocols \cite{angluin2005stably,berenbrink2018simple,bilke2017population,cai2012prove,izumi2015space,sudo2018loosely}, counting protocols \cite{aspnes2017time,beauquier2015space,beauquier2007self}, majority protocols \cite{angluin2008simple,gasieniec2017deterministic}, naming protocols \cite{burman2018brief}, and so on.

In this paper, we study a uniform $k$-partition problem, which divides a population into $k$ groups of the same size, where $k$ is a given positive integer.
The uniform $k$-partition problem has some applications. For example, we can save the battery by switching on only some groups. 
Another example is to execute multiple tasks by assigning different tasks to each group simultaneously.
Protocols for the uniform k-partition problem can be used to attain fault-tolerance \cite{delporte2006birds}.

As a previous work, Yasumi et al. \cite{bipartition,yasumi2019space} studied space complexity of uniform partition when the number of partitions is two (called uniform bipartition). In the paper, they considered four types of assumptions: 1) an initialized base station (BS) or no BS, 2) designated or arbitrary initial states of agents, 3) asymmetric or symmetric protocols, and 4) global or weak fairness.
A BS is a special agent that is distinguishable from other agents and has powerful capability.
An initialized BS means that the BS has a designated initial state in the initial configuration.
The BS enables us to construct efficient protocols, though it is sometimes difficult to implement.
	The assumption of initial states bear on the requirement of initialization and the fault-tolerant property.
If a protocol requires designated initial states, it is necessary to initialize all agents to execute the protocol.
Alternatively, if a protocol solves the problem with arbitrary initial states, we do not need to initialize agents other than the BS. 
This implies that, when agents transit to arbitrary states by transient faults, the protocol can reach the desired configuration by initializing the BS.
Symmetry of protocols is related to the power of symmetry breaking in the population.
 	Asymmetric protocols may include asymmetric transitions that make agents with the same states transit to different states. 
This needs a mechanism to break symmetry among agents and its implementation is not easy with heavily limited devices. 
Symmetric protocols do not include such asymmetric transitions.
Fairness is an assumption of interaction patterns. Though weak fairness guarantees only that every pair of agents interact infinitely often, global fairness makes a stronger assumption on the order of interactions.

For most combinations of assumptions, Yasumi et al. \cite{bipartition} clarified the solvability of the uniform bipartition problem and the minimum number of states to solve the problem. 
Tables \ref{tab:global} and \ref{tab:weak} show the solvability of the uniform bipartition.
These tables show the number of states to solve the uniform bipartition problem under various assumptions, where $P$ is the known upper bound of the number of agents.
The remaining case for an initialized BS and no BS is a protocol with an initialized BS and arbitrary initial states under weak fairness. For this case, they proved only that $P-2$ states are necessary.
In this paper, we will give tight lower and upper bounds of the number of states for this case.
In addition, recently Burman et al. \cite{burman2018brief} have considered the case with a non-initialized BS, which is distinguished from other agents but has an arbitrary initial state, for a naming problem. 
Because Yasumi et al. \cite{bipartition} did not consider the case, we also consider the case in this paper.

For the general case of an arbitrary number of partitions, Yasumi et al. \cite{bipartition} proposed a symmetric protocol with no BS and designated initial sates under global fairness. The protocol uses $3k-2$ states for an agent to construct $k$ groups of the same size. However, no protocol has been proposed for other combinations of assumptions.

 \subsection{Our Contributions}
 
\begin{table}[t]
\begin{minipage}{1.0\hsize}
\caption{The minimum number of states to solve the uniform $k$-partition problem.}
\begin{center}
\label{tab:weak2}
\begin{tabular}{|c|c|c|c|c|c|}
\hline
fairness & BS & initial states of agents & symmetry & upper bound & lower bound\\
\hline
\multirow{2}{*}{weak fairness} & \multirow{2}{*}{single} & \multirow{2}{*}{arbitrary} & asymmetric & $P$ & $P$\\
\cline{4-6}
 &  &  & symmetric & $P+1$ & $P+1$ \\
\hline 
 global fairness & no & designated & symmetric & $3k-2$ \cite{yasumi2019population} & $k$ (truism) \\
\hline
\end{tabular}
\end{center}
\end{minipage}
\end{table}

Our main contribution is to clarify the solvability of the uniform bipartition problem with arbitrary initial states under weak fairness in the model with an initialized BS.
A previous work \cite{bipartition} proved only that $P-2$ states are necessary for each agent, where $P$ is the known upper bound of the number of agents.
In this paper, we improve the lower bound from $P-2$ states to $P$ states for asymmetric protocols and from $P-2$ states to $P+1$ states for symmetric protocols.
Additionally, we propose an asymmetric protocol with $P$ states, and obtain a symmetric protocol with $P+1$ states by a scheme proposed in \cite{beauquier2015space}.

Another contribution is to clarify the solvability in case of a non-initialized BS for the uniform bipartition problem. 
For designated initial states, the protocol with an initialized BS, which is proposed in \cite{bipartition} can still work even if the BS is non-initialized.
In this paper, we prove the impossibility with arbitrary initial states in case of non-initialized BS.

By combining these results with the previous work \cite{bipartition}, we have clarified the tight upper and lower bounds on the number of states for an agent to solve the uniform bipartition problem for all combinations of assumptions (see Tables \ref{tab:global} and \ref{tab:weak}).

For the case of an initialized BS, arbitrary initial states, and weak fairness, it is interesting to compare these results with those of naming protocols \cite{burman2018brief}. A naming protocol aims to assign different states to all agents, and hence it can be regarded as a uniform $P$-partition protocol (the size of each group is zero or one). Burman et al. \cite{burman2018brief} prove that, to realize naming protocols in the same setting, $P$ states are necessary and sufficient for asymmetric protocols and $P+1$ states are necessary and sufficient for symmetric protocols. That is, naming protocols have the same space complexity as uniform $k$-partition protocols. Clearly naming protocols (or uniform $P$-partition protocols) require $P$ states to assign different states to $P$ agents. Interestingly uniform bipartition protocols still require $P$ states in this setting. On the other hand, the number of states is reduced to three or four when we assume designated initial states or global fairness.

Protocols proposed in this paper are available for the uniform $k$-partition problem, where $k$ is a given integer. That is, $P$ states and $P+1$ states are sufficient to realize asymmetric and symmetric protocols, respectively, to solve the uniform $k$-partition problem from arbitrary initial states under weak fairness in the model with an initialized BS.
Since the uniform bipartition is a special case of the uniform $k$-partition, the lower bound for the uniform bipartition problem can be applied to the uniform $k$-partition problem. That is, $P$ states and $P+1$ states are necessary to realize asymmetric and symmetric protocols, respectively, under the assumption.
Therefore, we have clarified the tight upper and lower bounds of the number of states for the uniform $k$-partition problem under the assumption (see Table \ref{tab:weak2}).

 \subsection{Related Works}
 The population protocol model was first introduced in \cite{angluin2006computation,angluin2007computational}.
In those papers, the class of computable predicates in this model was studied.
After that, many fundamental tasks have been studied such as leader election, counting, and majority.
Those problems have been studied under various assumptions such as existence of a base station, fairness, symmetry of protocols, and initial states of agents.
Many researchers have considered the leader election problem for both designated and arbitrary initial states. For designated initial states, leader election protocols have been studied intensively to minimize the time and space complexity \cite{alistarh2017time,alistarh2015polylogarithmic,berenbrink2018simple,bilke2017population,doty2018stable,gkasieniec2018almost,gkasieniec2018fast,sudo2018logarithmic}. 
Alistarh et al. \cite{alistarh2015polylogarithmic} proposed an algorithm that solves the problem in polylogarithmic stabilization time with polylogarithmic states.
In \cite{doty2018stable}, it was clarified that $\Omega(n)$ parallel time is necessary (i.e., $\Omega(n^2)$ interactions are necessary) to solve the problem with probability 1.
After that, many researchers focused on solving the problem with high probability and shrink the time and space complexity \cite{berenbrink2018simple,bilke2017population,gkasieniec2018almost,gkasieniec2018fast,sudo2018logarithmic}. 
On the other hand, for arbitrary initial states, self-stabilizing and loosely-stabilizing protocols are proposed \cite{angluin2005self,cai2012prove,izumi2015space,sudo2018loosely}. 
The counting problem, which aims to count the number of agents in the population, was introduced by \cite{beauquier2007self}.
After that, some researchers have studied the protocol to minimize the space complexity of the counting protocols \cite{beauquier2015space,izumi2014space2}. 
In \cite{aspnes2017time}, a time and space optimal protocol was proposed.
The majority problem is also a fundamental problem that aims to decide majority of initial states.  For the majority problem, many protocols have been studied \cite{alistarh2017time,alistarh2018space,alistarh2015fast,angluin2008simple,bilke2017population,gasieniec2017deterministic}. 
Although there are some difference in the model (existence of failure, deterministic or probabilistic solution, and so on), these works also aim to minimize the time and space complexity. 
Moreover, in recent years, Burman et al. \cite{burman2018brief} proposed a naming protocol which assigns a different state (called name) to each agent. In the paper, they completely clarify the solvability of the naming protocol under various assumptions.

The uniform $k$-partition problem and a similar problem have been considered in \cite{burman2018brief,lamani2016realization,tomoki2018differentiation,yasumi2019population}.
Lamani et al. \cite{lamani2016realization} studied a group composition problem, which aims to divide a population into groups of designated sizes. 
 They assume that half of agents make interactions at the same time and that all agents know $ n $. Therefore the protocol does not work in our setting. 
 In \cite{yasumi2019population}, Yasumi et al. proposed a uniform $k$-partition protocol that requires $3k-2$ states without the BS under global fairness.
Moreover, some of the authors extended the result of \cite{yasumi2019population} to the $R$-generalized partition problem, where the protocol divides all agents into $k$ groups whose sizes follow a given ratio $R$ \cite{tomoki2018differentiation}.
Since they assume designated initial states and global fairness, the protocol does not work in our setting. 
In addition, Delporte-Gallet et al. \cite{delporte2006birds} proposed a protocol solving the $k$-partition problem with less uniformity. 
This protocol guarantees that each group includes at least $n/(2k)$ agents, where $n$ is the number of agents. This protocol requires $k(k + 3)/2$ states under global fairness. 

  \section{Definitions}
 \subsection{Population Protocol Model}
A population is a collection of pairwise interacting agents, denoted by $A$.
A protocol ${\cal P}(Q, \delta)$ consists of $Q$ and $\delta$, where $Q$ is a set of possible states of agents and $\delta$ is a set of transitions on $ Q $.
Each transition in $\delta$ is denoted by $ (p, q) \rightarrow (p ', q') $,
which means that, when an agent with state $ p $ and an agent with state $ q $ interact, they transit their states to $ p'$ and $ q' $, respectively.
Transition $(p,q) \rightarrow (p',q')$ is null if both $p=p'$ and $q=q'$ hold. We omit null transitions in descriptions of algorithms.
Transition $ (p, q) \rightarrow (p ', q') $ is asymmetric if both $p=q$ and $p'\neq q'$ hold; otherwise, the transition is symmetric. Protocol ${\cal P}(Q, \delta)$ is symmetric if every transition in $\delta$ is symmetric. 
Protocol ${\cal P}(Q,\delta)$ is asymmetric if every transition in $\delta$ is symmetric or asymmetric.
Protocol ${\cal P}(Q,\delta)$ is deterministic if, for any pair of states $(p,q)\in Q\times Q$, exactly one transition $(p,q)\rightarrow (p', q') $ exists in $\delta$. We consider only deterministic protocols in this paper. 
 A global state of a population is called a configuration, defined as a vector of (local) states of all agents.
 A state of agent $a$ in configuration $C$, is denoted by $s(a, C)$. 
Moreover, when $C$ is clear from the context, we simply denote $s(a)$.
 Transition of configurations is described in the form $C \rightarrow C'$, which means that configuration $C'$ is obtained from $C$ by a single transition of a pair of agents.
 For configurations $C$ and $C'$, if there exists a sequence of configurations $C = C_0, C_1, \ldots , C_m = C'$ such that $C_i \rightarrow C_{i+1}$ holds for any $i$ ($0 \le i < m$), we say 
 $C'$ is reachable from $C$, denoted by $C \xrightarrow{*} C'$.
 An infinite sequence of configurations $E=C_0, C_1, C_2, \ldots$ is an execution of a protocol if $C_i \rightarrow C_{i+1}$ holds for any $i$ ($i \ge 0$).
An execution $E$ is weakly-fair if every pair of agents $a$ and $a'$ interacts infinitely often.
An execution segment is a subsequence of an execution.

In this paper, we assume that a single BS exists in $A$.
The BS is distinguishable from other non-BS agents, although non-BS agents cannot be distinguished.
That is, state set $ Q $ is divided into state set $ Q_b $ of a BS and state set $ Q_p $ of non-BS agents.
The BS has unlimited resources, in contrast with resource-limited non-BS agents.
 That is, we focus on the number of states $|Q_p|$ for non-BS agents and do not care the number of states $|Q_b|$ for the BS.
For this reason, we say a protocol uses $x$ states if $|Q_p|=x$ holds.
Throughout the paper, we assume that non-BS agents have arbitrary initial states. On the other hand, as for the BS, we consider two cases, an initialized BS and a non-initialized BS.
When we assume an initialized BS, the BS has a designated initial state while all non-BS agents have arbitrary initial states.
When we assume a non-initialized BS, the BS also has an arbitrary initial state.
 For simplicity, we use agents only to refer to non-BS agents in the following sections. To refer to the BS, we always use the BS (not an agent).
In the initial configuration, the BS and non-BS agents do not know the number of agents, but they know the upper bound $P$ of the number of agents.

 \subsection{Uniform $k$-Partition Problem} 
 Let $A_p$ be a set of all non-BS agents.
 Let $ f: Q_p \rightarrow \{color_1, color_2, \ldots, color_k \}$ be a function that maps a state of a non-BS agent to $ color_i (1 \leq i \leq k)$. We define a color of $a\in A_p$ as $f(s(a))$. We say agent $a\in A_p$ belongs to the $i$-th group if $ f(s(a))=color_i $ holds.

 Configuration $ C $ is stable if there is a partition $\{ G_1 $, $ G_2 $, $\ldots, G_k \}$ of $ A_p $ that satisfies the following condition:
 
 \begin{enumerate}
 	\item $\left||G_i|-|G_j|\right| \leq 1$ for any $i$ and $j$, and
 	\item For all $C^*$ such that $C \xrightarrow{*} C^*$, each agent in $G_i$ belongs to the $i$-th group at $C^*$ (i.e., at $C^*$, any agent $a$ in $G_i$ satisfies $ f(s(a))=color_i $).
 \end{enumerate}
 
  An execution $E = C_0$, $C_1$, $C_2$, $\ldots$ solves the uniform $ k $-partition problem if $E$ includes a stable configuration $C_t$.
  If every weakly-fair execution $ E $ of protocol $ {\cal P} $ solves the uniform $ k $-partition problem, we say protocol $ {\cal P} $ solves the uniform $ k $-partition problem under weak fairness.

\section{Impossibility Results for Initialized BS and Weak Fairness}
In this section, we give impossibility results of asymmetric and symmetric protocols for the uniform bipartition problem (i.e., $k=2$). Clearly these impossibility results can be applied to the uniform $k$-partition problem for $k>2$.
Recall that, for an initialized BS, we assume weak fairness and arbitrary initial states.

Since we consider the case of $k=2$, function $f$ is defined as $f : Q_p \rightarrow \{color_1,color_2\}$. In this section, we assign colors $red$ and $blue$ to $color_1$ and $color_2$, respectively, and we define $f$ as function $f : Q_p \rightarrow \{red,blue\}$ that maps a state of a non-BS agent to $red$ or $blue$. We say agent $a\in A_p$ is $red$ (resp., $blue$) if $f(s(a))=red$ (resp., $f(s(a))=blue$) holds.
We say $s$ is a $c$-state if $f(s)=c$ holds.
For $c\in\{red,blue\}$, we define $c$-agent as an agent that has a $c$-state. 
We define $\overline{red}=blue$ and $\overline{blue}=red$.

\subsection{Common Properties of Asymmetric and Symmetric Protocols}

First, we show basic properties that hold for both asymmetric and symmetric protocols.
The proofs of those properties are given in Appendix \ref{app:com}.
 Let $Alg$ be a protocol that solves the uniform bipartition. Recall that $P$ is the known upper bound of the number of agents. Hence, $Alg$ must solve the uniform bipartition when the actual number of agents is at most $P$. In the remainder of this subsection, we consider the case that the actual number of agents is $P-2$. 

Lemma \ref{lem:P-2blue} shows that, in any execution for $P-2$ agents, eventually all agents continue to keep different states. 

\begin{lemma}
\label{lem:P-2blue}
In any weakly-fair execution $E=C_0$, $C_1$, $C_2$, $\ldots$ of $Alg$ with $P-2$ agents and an initialized BS, there exists a configuration $C_h$ such that 1) $C_h$ is a stable configuration, and, 2) all agents have different states at $C_{h'}$ for any $h'\ge h$.
\end{lemma}

\begin{proof}$(Sketch)$
For contradiction, we assume that there exist two agents with the same state $s$ in a stable configuration of some execution $E$ with $P-2$ agents.
Next, consider an execution with $P$ agents such that two additional agents have $s$ as their initial states and other agents behave similarly to $E$.
In the execution, two additional agents do not join the interactions until $P-2$ agents converge to a stable configuration in $E$. At that time, two of the $P-2$ agents have state $s$ and additional two agents also have state $s$. 
We can prove that, from this configuration, $P-2$ agents cannot recognize the two additional agents and hence they make the same behavior as in $E$.
In addition, the two additional agents can keep state $s$.
Since the numbers of $red$ and $blue$ agents are balanced without the two additional agents and the two additional agents have the same state, the uniform bipartition problem cannot be solved.
This is a contradiction.
\end{proof}

In the next lemma, we prove that there exists a configuration $C$ such that, in any configuration reachable from $C$, all agents have different states. In addition, we also show that the system reaches $C$ in some execution.
\begin{definition}
\label{def:strcon}
Configuration $C$ is strongly-stable if 1) $C$ is stable, and, 2) for any configuration $C'$ with $C \xrightarrow{*} C'$, all agents have different states at $C'$.
\end{definition}

\begin{lemma}
\label{lem:strcon}
When the number of agents other than the BS is $P-2$, there exists an execution of $Alg$ that includes a strongly-stable configuration.
\end{lemma}

\begin{proof}$(Sketch)$
For contradiction, we assume that such execution does not exist.
First, consider a weakly-fair execution $E$ of $Alg$. By Lemma \ref{lem:P-2blue}, after some configuration $C_t$ in $E$, all agents have different states. From the assumption, $C_t$ is not strongly-stable. That is, there exists a configuration $C_u$ reachable from $C_t$ such that two agents have the same state.
Hence, we can construct another weakly-fair execution $E'$ of $Alg$ such that $E'$ is similar to $E$ until $C_t$ and $C_u$ occurs after that. 
By Lemma \ref{lem:P-2blue}, after some configuration $C_{t'}$ in $E'$, all agents have different states. Observe that $C_{t'}$ occurs after $C_u$.
From the assumption, there exists a configuration $C_{u'}$ reachable from $C_{t'}$ such that two agents have the same state.
Hence, similarly to $E'$, we can construct another weakly-fair execution $E''$ of $Alg$ such that $E''$ is similar to $E'$ until $C_{t'}$ and $C_{u'}$ occurs after that. 
By repeating this construction, we can construct a weakly-fair execution such that two agents have the same state infinitely often.
From Lemma \ref{lem:P-2blue}, this is a contradiction.
\end{proof}

\subsection{Impossibility of Asymmetric Protocols}
Here we show the impossibility of asymmetric protocols with $P-1$ states.

\begin{theorem}
In the model with an initialized BS, there is no asymmetric protocol that solves the uniform bipartition problem with $P-1$ states from arbitrary initial states under weak fairness, if $P$ is an even integer.
\label{the:P-1}
\end{theorem}

To prove the theorem by contradiction, we assume that such protocol $Alg_{asym}$ exists.
Let $Q_p=\{s_1$, $s_2$, $\ldots$, $s_{P-1}\}$ be a state set of agents other than the BS.
Let $Q_{blue}=\{ s \in Q_p \mid f(s)=blue\}$ be a set of blue states and $Q_{red}=\{ s \in Q_p \mid f(s)=red\}$ be a set of red states.
Without loss of generality, we assume that $|Q_{blue}| < |Q_{red}|$ holds.
Recall that Lemmas \ref{lem:P-2blue} and \ref{lem:strcon} can be applied to both symmetric and asymmetric algorithms.
Hence, the properties of the lemmas hold even in $Alg_{asym}$.
In this proof, based on the properties, we construct an execution of $P$ agents such that the BS does not recognize the difference from the execution of $P-2$ agents.
We show contradiction by proving that this execution does not achieve uniform bipartition.

By Lemma \ref{lem:P-2blue}, clearly $Alg_{asym}$ requires $P/2-1$ $blue$ states and $P/2 - 1$ $red$ states.
Consequently, we have the following two corollaries.
\begin{corollary}
\label{cor:P-2-2}
$|Q_{blue}|= P/2-1$ and $|Q_{red}| = P/2$ hold.
\end{corollary}

\begin{corollary}
\label{cor:Sblue}
For any weakly-fair execution of $Alg_{asym}$ with $P-2$ agents and an initialized BS, any strongly-stable configuration includes all states in $Q_{blue}$.
\end{corollary}

To prove the main theorem, we focus on the following weakly-fair execution of $Alg_{asym}$ with $P-2$ agents.
\begin{definition}
\label{def:E}
Consider a population $A=\{a_0, a_1, \ldots, a_{P-2}\}$ of $P-2$ agents and an initialized BS, where $a_0$ is the BS.
We define $E_\alpha = C_0, C_1, C_2, \ldots$ as a weakly-fair execution of $Alg_{asym}$ for population $A$ that satisfies the following conditions.
\begin{itemize}
\item $E_\alpha$ includes a strongly-stable configuration $C_t$, and, 
\item for any $u \ge 0$, agents that interact at $C_{t+2u} \rightarrow C_{t+2u+1}$ also interact at $C_{t+2u+1} \rightarrow C_{t+2(u+1)}$.
\end{itemize}
\end{definition}

Note that, in $E_\alpha$, the system reaches a strongly-stable configuration $C_t$ (this is possible from Lemma \ref{lem:strcon}), and after $C_t$ agents always repeat the same interaction twice.
\begin{definition}
\label{def:Qt}
We define $Q_t$ as a set of states that appear after $C_t$ in $E_\alpha$.
\end{definition}

Note that, since $C_t$ is strongly-stable, $Q_t$ includes at least $P-2$ states. This implies that $Q_t$ includes all states in $Q_p$ or does not include one state in $Q_p$. From Corollary \ref{cor:Sblue}, $Q_{blue} \subset Q_t$ holds.

The following lemmas give key properties of $Alg_{asym}$ to prove Theorem \ref{the:P-1}.
We will present proofs of these lemmas later.

\begin{lemma}
\label{lem:P-2blue2}
For any distinct states $p$ and $q$ ($p\neq q$) such that $p\in Q_{blue}$ and $q \in Q_{t}$ hold, transition rule $(p,q) \rightarrow (p',q')$ satisfies the following conditions.

\begin{itemize}
\item If $q\in Q_{red}$ or $q\in Q_b$ (i.e., $q$ is a state of the BS) holds, $p' = p$ holds.
\item If $q\in Q_{blue}$ holds, either $(p', q')=(p, q)$ or $(p', q')=(q, p)$ holds.
\end{itemize}
\end{lemma}

\begin{lemma}
\label{lem:Q^+}
There is a non-empty state set $Q^* \subseteq Q_{blue}$ that satisfies the following conditions.

\begin{itemize}
\item For any state $p \in Q^*$, transition rule $(p,p) \rightarrow (p',q')$ satisfies $p'\in Q^*$ and $q'\in Q^*$.
\item Assume that, in a configuration $C$, there exists a subset of agents $A^*$ such that all agents in $A^*$ have states in $Q^*$ and $|A^*|=|Q^*|+1$ holds. In this case, for any agent $a_r \in A^*$ and any state $q \in Q^*$, there exists an execution segment such that 1) the execution segment starts from $C$, 2) $a_r$ has state $q$ at the last configuration, 3) only agents in $A^*$ join interactions, and 4) all agents in $A^*$ have states in $Q^*$ at the last configuration.
\end{itemize}
\end{lemma}

Lemma \ref{lem:Q^+} means that, if $|Q^*|+1$ agents have states in $Q^*$, we can make an arbitrary agent with a state in $Q^*$ transit to an arbitrary state in $Q^*$.
Using these lemmas, we show the theorem by constructing a weakly-fair execution of $Alg_{asym}$ with $P$ agents that cannot be distinguished from execution $E_\alpha$.

\subsection*{Proof of Theorem \ref{the:P-1}}

Consider a population $A'=\{a'_0$, $\ldots$, $a'_{P}\}$ of $P$ agents and an initialized BS, where $a'_0$ is the BS.
Let $C'_0$ be an initial configuration such that initial states of $a'_0$, $\ldots$, $a'_{P}$
are $s(a_0, C_0)$, $\ldots$, $s(a_{P-2}, C_0)$, $s^*$, $s^*$, where $s^*$ is a state in $Q^*$.

For $A'$ we construct an execution $E_\beta=C'_0$, $C'_1$, $\ldots$, $C'_{t}$, $\ldots$ using execution $E_\alpha$ as follows.

\begin{itemize}
\item For $0 \le u \le t-1$, when $a_i$ and $a_j$ interact at $C_u \rightarrow C_{u+1}$ in $E_\alpha$, $a'_i$ and $a'_j$ interact at $C'_u\rightarrow C'_{u+1}$ in $E_\beta$.
\end{itemize}
Clearly, $s(a'_i,C'_t)=s(a_i,C_t)$ holds for any $i$ ($0 \le i \le P-2$).
Since $s(a'_{P-1},C'_t) = s(a'_{P},C'_t)=s^*$ holds, the difference in the numbers of $red$ and $blue$ agents remains two and consequently $C'_{t}$ is not a stable configuration.

To construct the remainder of $E_\beta$, first let us consider the characteristics of $C'_t$. Let $A_q \subseteq A$ be a set of agents that have states in $Q^*$ at $C_t$, and let $\bar{A_q}=A-A_q$. Since all agents have different states and all states in $Q_{blue}$ appear in $C_t$ by Corollary \ref{cor:Sblue}, we have $|A_q|=|Q^*|$ from $Q^*\subseteq Q_{blue}$. Let $A'_q \subseteq A'$ be a set of agents that have states in $Q^*$ at $C'_t$, and let $\bar{A'_q}=A'-A'_q$. Note that, for $i\le P-2$, $a_i\in A_q$ holds if and only if $a'_i \in A'_q$ holds. Since $a'_{P-1}$ and $a'_P$ are also in $A'_q$, we have $|A'_q|=|Q^*|+2$. 
In the following, we construct the remainder of execution $E_\beta$ that includes infinitely many configurations similar to $E_\alpha$.
we define similarity of configurations in $E_\beta$ and $E_\alpha$ as follows.

\begin{definition}
\label{def:sim}
We say configuration $C'_u$ ($u\ge t$) in $E_\beta$ is similar to $C_v$ ($v\ge t$) in $E_\alpha$ if the following conditions hold:
\begin{itemize}
\item For any agent $a_i\in A_q$, $s(a_i,C_v) \in Q^*$ holds.
\item For any agent $a'_i\in A'_q$, $s(a'_i,C'_u) \in Q^*$ holds.
\item For any agent $a'_i\in \bar{A'_q}$ (i.e., $a_i \in \bar{A_q}$), $s(a'_i,C'_u)=s(a_i,C_v)$ holds.
\end{itemize}
\end{definition}

Let us focus on an execution segment $e=C_{t+2u},C_{t+2u+1},C_{t+2(u+1)}$ of $E_\alpha$ for any $u\ge 0$, and consider a configuration $C'_x$ of $A'$ such that $C'_x$ is similar to $C_{t+2u}$. From now, we explain the way to construct an execution segment $e'=C'_x,\ldots,C'_y$ of $E_\beta$ that guarantees that $C'_y$ is similar to $C_{t+2(u+1)}$. Since $C'_t$ is similar to $C_t$, we can repeatedly apply this construction and construct an infinite execution $E_\beta$. As a result, for any $u\ge 0$, $E_\beta$ includes a configuration $C'$ that is similar to $C_{t+2u}$. Since $C'$ includes $P-1$ $red$ agents and $P+1$ $blue$ agents, $E_\beta$ does not include a stable configuration. Note that $E_\beta$ is not necessarily weakly-fair, but later we explain the way to construct a weakly-fair execution from $E_\beta$.

Consider configuration $C'_x$ that is similar to $C_{t+2u}$. Assume that, in $E_\alpha$, agents $a_i$ and $a_j$ interact in $C_{t+2u} \rightarrow C_{t+2u+1}$. Recall that $a_i$ and $a_j$ also interact in $C_{t+2u+1} \rightarrow C_{t+2(u+1)}$. We construct execution segment $e'$ as follows:

\begin{itemize}
\item Case that $a_i \in A_q \wedge a_j \in A_q$ holds. Since $s(a_i,C_{t+2u}) \in Q^* \subseteq Q_{blue}$ and $s(a_j,C_{t+2u}) \in Q^* \subseteq Q_{blue}$ hold, $s(a_i,C_{t+2(u+1)}) \in Q^*$ and $s(a_j,C_{t+2(u+1)}) \in Q^*$ also hold from Lemma \ref{lem:Q^+} (the first condition) and Lemma \ref{lem:P-2blue2}. Since other agents do not change their states, $C'_x$ is similar to $C_{t+2(u+1)}$. Hence, in this case, we consider that the constructed execution segment $e'$ is empty. 

\item Case that either $a_i \in A_q \wedge a_j \in \bar{A_q}$ or $a_i \in \bar{A_q} \wedge a_j \in A_q$ holds. Without loss of generality, we assume that $a_i \in A_q \wedge a_j \in \bar{A_q}$ holds. In this case, $s(a'_i,C'_x) \in Q^*$ is not necessarily equal to $\alpha=s(a_i,C_{t+2u}) \in Q^*$. Hence, in the execution segment $e'$, we first make some agent $a'_r \in A'_q$ enter state $\alpha$ by interactions among agents in $A'_q$. By Lemma \ref{lem:Q^+} (the second condition) and $|A'_q|=|Q^*|+2$, such interactions exist and all agents in $A'_q$ have states in $Q^*$ after the interactions. Let $C'_z$ be the resultant configuration. Clearly $C'_z$ is similar to $C_{t+2u}$ and $s(a'_r,C'_z)=s(a_i,C_{t+2u}) \wedge s(a'_j,C'_z)=s(a_j,C_{t+2u})$ holds. After that, $a'_r$ and $a'_j$ interact twice. We regard the resultant configuration as $C'_y$ (i.e., the last configuration of the constructed execution segment $e'$). Clearly both $s(a'_r,C'_y)=s(a_i,C_{t+2(u+1)})$ and $s(a'_j,C'_y)=s(a_j,C_{t+2(u+1)})$ hold. Since $C'_z$ is similar to $C_{t+2u}$ and $s(a'_j,C'_y)=s(a_j,C_{t+2(u+1)})$, it is sufficient to prove $s(a'_r,C'_y) \in Q^*$ to guarantee that $C'_y$ is similar to $C_{t+2(u+1)}$. Observe that $s(a_j,C_{t+2u}) \notin Q^*$. This is because, since $C_{t+2u}$ is strongly-stable, all agents have different states and agents in $A_q$ occupy all states in $Q^*$ (the first condition of Definition \ref{def:sim}). Hence, $s(a'_r,C'_z)=s(a_i,C_{t+2u})\in Q^*$ is not equal to $s(a'_j,C'_z)=s(a_j,C_{t+2u})\notin Q^*$. Consequently, from $s(a'_r,C'_z)\in Q^*\subseteq Q_{blue}$, $s(a'_r,C'_y)=s(a'_r,C'_z)\in Q^*$ holds by Lemma \ref{lem:P-2blue2}. Therefore, $C'_y$ is similar to $C_{t+2(u+1)}$.

\item Case that $a_i \in \bar{A_q} \wedge a_j \in \bar{A_q}$ holds. In this case, since $s(a'_i,C'_x)=s(a_i,C_{t+2u})$ and $s(a'_j,C'_x)=s(a_j,C_{t+2u})$ hold, $a'_i$ and $a'_j$ simply interact twice consecutively. We regard the resultant configuration as $C'_y$ (i.e., the last configuration of the constructed execution segment $e'$). Clearly, since $a'_i$ and $a'_j$ change their states similarly to $a_i$ and $a_j$, $C'_y$ is similar to $C_{t+2(u+1)}$.
\end{itemize}

Now we have constructed infinite execution $E_\beta$, but $E_\beta$ is not necessarily weakly-fair. In the following, we construct a weakly-fair execution $E_\gamma$ of population $A'$ by slightly modifying $E_\beta$. To guarantee that $E_\gamma$ is weakly-fair, for any pair of agents $(a'_i,a'_j)$, $a'_i$ and $a'_j$ should interact infinite number of times in $E_\gamma$. For pair of agents $(a'_i,a'_j)$ with $a'_i \in \bar{A'_q}$ and $a'_j \in \bar{A'_q}$, $a'_i$ and $a'_j$ interact infinite number of times in $E_\beta$ because $E_\alpha$ is weakly-fair and $a'_i$ interacts with $a'_j$ in $E_\beta$ when $a_i$ interacts with $a_j$ in $E_\alpha$. For pair of agents $(a'_i,a'_j)$ with $a'_i \in A'_q$ and $a'_j \in A'_q$, we can arbitrarily add interactions of them because, by Lemma \ref{lem:Q^+} (the first condition) and Lemma \ref{lem:P-2blue2} (the second condition), they keep their states in $Q^*$ and consequently do not influence similarity of configurations.

Hence, we consider the remaining pair $(a'_i,a'_j)$, that is, either $a'_i \in A'_q \wedge a'_j \in \bar{A'_q}$ or $a'_i \in \bar{A'_q} \wedge a'_j \in A'_q$ holds. Without loss of generality, we assume that $a'_i$ is in $A'_q$ and $a'_j$ is in $\bar{A'_q}$. Since $E_\alpha$ is weakly-fair, $a_j$ interacts with an agent in $A_q$ infinite number of times in $E_\alpha$. Recall that these interactions correspond to interactions of $a'_j$ and $a'_r$ in $E_\beta$, and $a'_r$ can be arbitrarily selected from $A'_q$. For this reason, we can choose $a'_r$ in a round-robin manner so that $a'_j$ interacts with any agent in $A'_q$ infinite number of times. For example, when $a_j$ and an agent in $A_q$ first interact (after $C_t$), we choose an agent in $A'_q$ as $a'_r$, and then in the next interaction of $a_j$ and an agent in $A_q$ we can choose another agent in $A'_q$ as $a'_r$. By this construction, $a'_j$ can interact with any agent $a'_i$ in $A'_q$ infinite number of times.

From this way, we can construct a weakly-fair execution $E_\gamma$ similarly to $E_\beta$. However, for any $u\ge 0$, $E_\gamma$ includes a configuration $C''$ that is similar to $C_{t+2u}$. Since $C''$ includes $P-1$ $red$ agents and $P+1$ $blue$ agents, $E_\gamma$ does not include a stable configuration. This is a contradiction.

\subsection*{The Proof Sketch of Lemma \ref{lem:P-2blue2}}
Consider the case that transition $(p,q) \rightarrow (p',q')$ occurs at a strongly-stable configuration with $P-2$ agents.
By Corollary \ref{cor:Sblue}, since any strongly-stable configuration includes all states in $Q_{blue}$, $(p,q) \rightarrow (p',q')$ can occur at the configuration.

First, consider the case that $q\in Q_{red}$ or $q\in Q_b$ holds.
For contradiction, assume that $p' \neq p$ holds.
By Corollary \ref{cor:Sblue}, since an agent with $p'$ exists in the strongly-stable configuration, two agents with $p'$ exist after transition $(p,q) \rightarrow (p',q')$.
By the definition of strongly-stable configuration, this is a contradiction.

Next, consider the case that $q\in Q_{blue}$ holds.
For contradiction, assume that $(p', q') \neq (p, q)$ and $(p', q') \neq (q, p)$ hold.
By the definition of stable configuration, $p'$ and $q'$ are $blue$.
Hence, by Corollary \ref{cor:Sblue}, since an agent with any state in $Q_{blue}$ exists in the strongly-stable configuration, two agents with the same state in $Q_{blue}$ exist after transition $(p,q) \rightarrow (p',q')$.
By the definition of strongly-stable configuration, this is a contradiction.

\subsection*{The Proof Sketch of Lemma \ref{lem:Q^+}}
First, to show the proof sketch, we give some definitions.

\begin{definition}
\label{def:rea}
For states $q$ and $q'$, we say $q \rightsquigarrow q'$ if there exists a sequence of states $q = q_0, q_1, \cdots , q_k = q'$ such that, for any $i(0 \le i < k)$, 
transition rule $(q_i,q_i) \rightarrow (q_{i+1},x_i)$ or $(q_i,q_i) \rightarrow (x_i, q_{i+1})$ exists for some $x_i$.
\end{definition}

\begin{definition}
\label{def:rea2}
For states $q$ and $q'$, we say $q \overset{*}{\rightsquigarrow} q'$ if $x \rightsquigarrow q'$ holds for any $x$ such that $q \rightsquigarrow x$ holds.
\end{definition}
Note that, in these definitions, we consider only interactions of agents with the same state.
We say two agents are homonyms if they have the same state.
Intuitively, $q \rightsquigarrow q'$ means  that an agent with state $q$ can transit to $q'$ by only interactions with homonyms.
Also, $q \overset{*}{\rightsquigarrow} q'$ means that, even if an agent with state $q$ transits to any state $x$ by interactions with homonyms, it can still transit from $x$ to $q'$ by interactions with homonyms.

In Appendix \ref{app:lem34}, we show that there exists $p^*$ such that $p^* \overset{*}{\rightsquigarrow} p^*$ holds.
Let $Q_{p^*}=\{ q \mid p^* \overset{*}{\rightsquigarrow} q \}$.
In this proof, we show that $Q_{p^*}$ satisfies the conditions of $Q^*$ of Lemma \ref{lem:Q^+}.
Clearly, if homonyms with states in $Q_{p*}$ interact, they transit to states in $Q_{p*}$. 
This implies that $Q_{p^*}$ satisfies the first condition of $Q^*$ of the lemma.
To prove the second condition, we first show that, when $|Q_{p*}|$ agents have states in $Q_{p*}$ initially, for any $s \in Q_{p*}$, there exists an execution such that only homonyms in the $|Q_{p*}|$ agents interact and eventually some agent transits to state $s$.
To show this, we define a potential function $\Phi(C, s)$ for configuration $C$ and state $s \in Q_{p*}$.
Intuitively, $\Phi(C, s)$ represents how far configuration $C$ is from a configuration that includes an agent with state $s$.
To define $\Phi(C, s)$, we define $\DtQ{s_i}{s}$ as follows.

\begin{definition}
\label{def:tradisq}
$\DtQ{s_i}{s}$ is a function that satisfies the following property.
\begin{itemize}
\item If $s_{i} = s $ holds, $\DtQ{s_i}{s}=0$ holds.
\item If $s_i \neq s$ and $s_i \in Q_{p*}$ holds, $\DtQ{s_i}{s}= \min\{\DtQ{s^1_j}{s}, \DtQ{s^2_j}{s}\}+1$ holds when transition rule $(s_i,s_i)\rightarrow (s^1_j,s^2_j)$ exists.
\item If $s_i \notin Q_{p^*}$ holds, $\DtQ{s_i}{s}= \infty$ holds.
\end{itemize}

\end{definition}

Intuitively, $\DtQ{s_i}{s}$ gives the minimum number of interactions to transit from state $s_i$ to state $s$ when allowing only interactions with homonyms. Note that, for any $s_i\in Q_{p*}$, $s_i$ can transit to $s$ because $s_i \rightsquigarrow p \rightsquigarrow s$ holds.

\begin{definition}
\label{def:trapro}
Consider configuration $C$ such that $z=|Q_{p^*}|$ agents $a_1$, $\ldots$, $a_z$ have states in $Q_{p^*}$. In this case, we define potential function $\Phi(C, s)$ as a multi set $\{ \DtQ{s(a_1, C)}{s}$, $\DtQ{s(a_2, C)}{s}$, $\DtQ{s(a_3, C)}{s}$, $\ldots$, $\DtQ{s(a_z, C)}{s} \}$.
\end{definition}

\begin{definition}
\label{def:trapro}
For distinct $\Phi(C_1, s)$ and $\Phi(C_2, s)$, we define a comparative operator of them as follows:
Let $i$ be the minimum integer such that the number of $i$-elements is different in $\Phi(C_1, s)$ and $\Phi(C_2, s)$. If the number of $i$-elements in $\Phi(C_1, s)$ is smaller than $\Phi(C_2, s)$, we say $\Phi(C_1, s) < \Phi(C_2, s)$. 
\end{definition}

From now, we show that there exists an execution such that some agent transits to $s$.
Let $C$ be a configuration with $|Q_{p^*}|$ agents such that all agents have states in $Q_{p*}$ and there does not exist an agent with $s$ in $C$.
Since $|Q_{p^*}|$ agents have states in $Q_{p*}$ in $C$ and there does not exist an agent with $s$ in $C$, there exist homonyms in $C$.
When homonyms with a state in $Q_{p*}$ interact, they transit to states in $Q_{p*}$.
These imply that, when homonyms interact at $C \rightarrow C'$, either an agent with $s$ or homonyms with a state in $Q_{p*}$ exist in $C'$.
Thus, for contradiction, assume that there exists an infinite execution segment $e=C_0$, $C_1$, $C_2$, $\ldots$ with $|Q_{p^*}|$ agents such that only homonyms interact and any agent never has $s$ in $e$, where $C_0$ is a configuration such that all agents have states in $Q_{p*}$.
For $e$, $\Phi(C_0, s) > \Phi(C_1, s) > \Phi(C_2, s) > \cdots$ holds.
This is because, since any $p \in Q_{p^*}$ satisfies $p \rightsquigarrow p^* \rightsquigarrow s$, $\DtQ{s(a,C_i)}{s} > \DtQ{s(a,C_{i+1})}{s} $ holds for at least one agent $a$ that interacts at $C_i \rightarrow C_{i+1}$.
Hence, eventually some agent has $s$ in $e$. By the definition of $e$, this is a contradiction.

From now, we prove the second condition of Lemma \ref{lem:Q^+}. Let $A^*$ be a set of agents such that $|A^*|=|Q_{p^*}|+1$, and assume that all agents in $A^*$ have states in $Q_{p^*}$.
The existence of the above execution implies that, for any agent $a_r \in A^*$, we can make some agent in $A^* - \{ a_r \}$ transits to state $s(a_r)$ by interactions among $ A^*  - \{ a_r\}$.
Then, we can make an interaction with homonyms between $a_r$ and an agent with $s(a_r)$.
After that, since $a_r$ has a state in $Q_{p^*}$, all agents in $A^*$ keep states in $Q_{p^*}$.
Hence, in the same way, by making interaction repeatedly between $a_r$ and an agent with $s(a_r)$, $a_r$ can transit to any $q \in Q_{p^*}$ because any $p \in Q_{p^*}$ satisfies $p \rightsquigarrow p^* \rightsquigarrow q$.
Therefore, $Q_{p^*}$ satisfies the second condition and thus the lemma holds.

\subsection{Impossibility of Symmetric Protocols}

In this subsection, we show the impossibility of symmetric protocols with $P$ states.
To prove this impossibility, we use ideas of the impossibility proof for the naming protocol \cite{burman2018brief}.
This work shows that, in the model with an initialized BS, there is no symmetric naming protocol with $P$ states from arbitrary initial states under weak fairness.
We apply the proof of \cite{burman2018brief} to the uniform bipartition but, since the treated problem is different, we need to make a non-trivial modification (the proof is presented in Appendix \ref{app:imp2}).

\begin{theorem}
\label{the:sym}
In the model with an initialized BS, there is no symmetric protocol that solves the uniform bipartition problem with $P$ states from arbitrary initial states under weak fairness, if $P$ is an even integer.
\end{theorem}

In the case of naming protocols \cite{burman2018brief}, the impossibility proof proves that a unique state (called sink state) always exists. However, in the case of uniform bipartition protocols, sometimes no sink state exists. To treat this situation, we additionally define a sink pair, which is a pair of two states that has a similar property of a sink state. We show that either a sink state or a sink pair exists, and, we prove that there is no symmetric protocol in both cases.

\section{Possibility Results for Initialized BS and Weak Fairness}
In this section, we propose both asymmetric and symmetric protocols for the uniform $k$-partition problem.
The asymmetric protocol requires $P$ states and the symmetric protocol requires $P+1$ states.
By impossibility results, these protocols are space-optimal.

\subsection{An Asymmetric Protocol}
In this subsection, we show a $P$-state asymmetric protocol for the uniform $k$-partition problem.
The idea of the protocol is to assign states $0$, $1$, $\ldots$, $n-1$ to $n$ agents one by one and then regard an agent with state $s$ as a member of the ($s \mod k$)-th group.
One may think that, to implement this idea, we can directly use a naming protocol \cite{burman2018brief}, where the naming protocol assigns different states to agents by using $P$ states if $n \le P$ holds.
Actually, if $n=P$ holds, the naming protocol assigns states $0$, $1$, $\ldots$, $P-1$ to $P$ agents one by one and hence it achieves uniform $k$-partition.
However, if $n<P$ holds, the naming protocol does not always achieve uniform $k$-partition. 
For example, in the case of $(n-1)k<P$, the naming protocol may assign states $0$, $k$, $2k$, $\ldots$, $(n-1)k$ to $n$ agents one by one, which implies that all agents are in the $0$-th group.
 
Algorithm \ref{alg1} shows a $P$-state asymmetric protocol for the uniform $k$-partition problem.
In the protocol, the BS assigns states $0$, $1$, $\ldots$, $n-1$ to $n$ agents one by one.
To do this, the BS maintains variable $M$, which represents the state the BS will assign next. The BS sets $M=0$ initially, and increments $M$ whenever it assigns $M$ to an agent.
Consider an interaction between the BS and an agent with state $x$. If $x$ is smaller than $M$, the BS judges that it has already assigned a state to the agent, and hence it does not update the state. If $x$ is $M$ or larger, the BS assigns state $M$ to the agent and increments $M$.
When the BS assigns state $x$ to an agent, there may exist another agent with state $x$ because of arbitrary initial states.
To treat this case, when two agents with the same state $x$ interact, one transits to state $x+1$ and the other keeps its state $x$. 
By repeating such interactions, eventually exactly one agent has state $x$.
By this behavior, the BS eventually assigns states $0$, $1$, $\ldots$, $n-1$ to $n$ agents one by one, and hence the algorithm achieves uniform $k$-partition.

\begin{algorithm}[t!]
	\caption{Asymmetric uniform $k$-partition protocol}         
	\label{alg1}
	\algblock{when}{End}
	\begin{algorithmic}[1]
		\renewcommand{\algorithmicrequire}{\textbf{A variable at BS}}
		\Require
		\Statex $M$: The state that the BS assigns next, initialized to 0

		\renewcommand{\algorithmicrequire}{\textbf{A variable at a mobile agent $a$:}}
		\Require
		\Statex $S_a \in \{0, 1, 2, \ldots, P-1\}$: The agent state, initialized arbitrarily. Agent $a$ belongs to the $(S_a\mod k)$-th group.

		\While {a mobile agent $a$ interacts with BS}
		\If{$M \le S_a$}
		\State $S_a = M$
		\State $M = M + 1$
		\EndIf
		\EndWhile

		\While {two mobile agent $a$ and $b$ interact}
		\If{$S_a = S_b$ and $S_a<P-1$}
		\State $S_a = S_a + 1$
		\EndIf
		\EndWhile		

	\end{algorithmic}
\end{algorithm}

As a result, we obtain the following theorem. The proof is in Appendix \ref{app:alg}.
\begin{theorem}
\label{the:kpart}
Algorithm \ref{alg1} solves the uniform $k$-partition problem.
This means that, in the model with an initialized BS, there exists an asymmetric protocol with $P$ states and arbitrary initial states that solves the uniform $k$-partition problem under weak fairness, where $P$ is the known upper bound of the number of agents.
\end{theorem}

\begin{remark}
Interestingly, when $P$ is odd, Algorithm \ref{alg1} solves the uniform bipartition even if the number of agent states is $P-1$.
Concretely, let $S_a \in \{ 1, 2, 3, \ldots, P-1 \}$ be a set of agent states, and initialize variable $M$ to $1$.
Then, Algorithm \ref{alg1} converges to a configuration such that there exist two agents with state $P-1$ (and other states are held by exactly one agent).
This is because, in the algorithm, the BS assigns $P-1$ agents to $P-1$ states one by one, and, since the algorithm works under weak fairness, the remaining one agent shifts its state until state $P-1$.
In the configuration, the difference in the numbers of $red$ and $blue$ agents is one. Moreover, every agent does not change its own state after the configuration.
Hence, the uniform bipartition is solved.
\qed
\end{remark}

\subsection{A Symmetric Protocol}
In this subsection, we propose a $(P+1)$-state symmetric protocol for the uniform $k$-partition problem.
We can easily obtain the protocol by a scheme proposed in \cite{beauquier2015space}.
In \cite{beauquier2015space}, a $P$-state symmetric protocol for the counting problem is proposed.
The counting protocol assigns different states in $\{1, \ldots, n\}$ to $n$ agents and keeps the configuration if $n < P$ holds. 
Hence, by regarding $P+1$ as the upper bound of the number of agents and allowing $P+1$ states, the protocol assigns different states in $\{1,...,n\}$ to $n$ agents for any $n\le P$. 
This implies that, as in the previous subsection, the protocol can achieve the uniform $k$-partition by regarding an agent with $x$ as a member of the $(x \mod k)$-th group.

\begin{theorem}
\label{the:name16}
In the model with an initialized BS, there exists a symmetric protocol with $P+1$ states and arbitrary initial states that solves the uniform $k$-partition problem under weak fairness, where $P$ is the known upper bound of the number of agents.
\end{theorem}

\section{Results for Non-initialized BS}
In this section, we show the impossibility with non-initialized BS.
In the proof, we use ideas of the impossibility proof for the uniform bipartition protocol \cite{bipartition}.
This work shows that, in the model with no BS, there is no protocol for uniform bipartition problem with arbitrary initial states under global fairness.

\begin{theorem}
	In the model with non-initialized BS, no protocol with arbitrary initial states solves the uniform bipartition problem under global fairness.
\end{theorem}

\begin{proof}
	For contradiction, we assume such a protocol $Alg$ exists.
	Moreover, we assume $n$ is even and at least 4. 
	We consider the following two cases.
	
	First, consider population $A=\{a_0,\ldots,a_n\}$ of $n$ agents and a non-initialized BS, where $a_0$ is the BS.
	For $A$, consider an execution $E=C_0, C_1, \cdots$ of $Alg$.
	From the definition of $Alg$, there exists a stable configuration $C_t$.
	Hence, both the number of $red$ agents and the number of $blue$ agents are $n/2$ at $C_t$.
	By the definition of a stale configuration, the color of agent $a_i$ (i.e.,$ f(s(a_i)) $) never changes for any $a_i$ ($1 \leq i \leq n$) after $C_t$ even if agents interact in any order.
	
	Next, consider population $A'= \{a'_0\} \cup \{a'_i | f(s(a_i,C_t)) = red \}$, where $a'_0$ is the BS. 
	For $A'$, consider an execution $E'=C'_0, C'_1,\cdots$ of $Alg$ from the initial configuration $ C'_0 $ such that $s(a'_i,C'_0)=s(a_i,C_t)$ holds for any $a'_i \in A'$.
	Note that, since we assume a non-initialized BS, even the BS can have $s(a_0,C_t)$ as its initial state.
	Since all agents are $red$ at $C'_0$, some agents must change their colors to reach a stable configuration.
	This implies that, after $C_t$ in execution $E$, agents change their colors if they interact similarly to $E'$. This is a contradiction.
\end{proof}

\section{Conclusion}
In this paper, we clarify solvability of the uniform bipartition with arbitrary initial states under weak fairness in the model with an initialized BS.
Concretely, for asymmetric protocols, we show that $P$ states are necessary and sufficient to solve the uniform $k$-partition problem under the assumption, where $P$ is the known upper bound of the number of agents.
For symmetric protocols, we show that $P+1$ states are necessary and sufficient under the assumption.
Moreover, these upper and lower bounds can be applied to the $k$-partition problem under the assumption.
There are some open problems as follows:
\begin{itemize}
\item Are there some relations between the uniform $k$-partition problem and other problems such as counting, leader election, and majority?

\item What is the time complexity of the uniform $k$-partition problem?

\item What is the space complexity of the approximately $k$-partition problem which divides a population into $k$ groups approximately?
For any $i$ and $j$, when the difference in the numbers of agents with group $i$ and agents with group $j$ is less than or equal to one, the problem is the same as the uniform $k$-partition problem.
On the other hand, when the difference is greater than one, the space complexity is not clarified.

\end{itemize}


%
%
%
%
\bibliographystyle{splncs04}
\bibliography{ref}

\section*{Appendix}
\appendix

\section{Proofs of Lemmas \ref{lem:P-2blue} and \ref{lem:strcon}}
\label{app:com}

\setcounter{lemma}{0}
\begin{lemma}
In any weakly-fair execution $E=C_0$, $C_1$, $C_2$, $\ldots$ of $Alg$ with $P-2$ agents and a single BS, there exists a configuration $C_h$ such that 1) $C_h$ is a stable configuration, and, 2) all agents have different states at $C_{h'}$ for any $h'\ge h$.
\end{lemma}

\begin{proof}
Assume, by way of contradiction, that such configuration $C_h$ does not exist.

First, consider a population $A=\{a_0$, $a_1$, $\ldots$, $a_{P-2}\}$ of $P-2$ agents and a single BS, where $a_0$ is the BS.
Let $E=C_0, C_1, \ldots, C_t, \ldots$ be a weakly-fair execution of $Alg$, where $C_t$ is a stable configuration.
Since $C_h$ does not exist, two agents have the same state at infinitely many configurations after $C_t$. Since the number of agents is finite and the number of states is finite, there exist $y$, $a_p$, and $a_{p'}$ such that configurations satisfying $y=s(a_p)=s(a_{p'})$ appear infinitely often after $C_t$.
	
Next, consider population $A'=\{a'_0$, $\ldots$, $a'_{P}\}$ of $P$ agents and a single BS, where $a'_0$ is a BS. 
We consider an initial configuration $C'_0$ such that initial states of $a'_0$, $\ldots$, $a'_{P}$ are $s(a_0, C_0)$, $\ldots$, $s(a_{P-2}, C_0)$, $y$, $y$, respectively.
For $A'$ we define an execution $E'=C'_0$, $C'_1$, $\cdots$, $C'_t$, $\cdots$ using execution $E$ as follow.
	
\begin{itemize}
\item For $0\le u \le t-1$, when $a_i$ and $a_j$ interact at $C_u\rightarrow C_{u+1}$, $a'_i$ and $a'_j$ interact at $C'_u\rightarrow C'_{u+1}$.
\end{itemize}
Clearly, $s(a'_i,C'_t)=s(a_i,C_t)$ holds for any $i$ ($0 \le i \le P-2$).
Since $s(a'_{P},C'_t)=s(a'_{P-1},C'_t)=y$ holds, the difference in the numbers of $red$ and $blue$ agents remains two and consequently $C'_t$ is not a stable configuration.
	
After $C'_t$, we define an execution as follows.
This definition aims to make $P-2$ agents behave similarly to $E$ and two agents keep state $y$.
\begin{itemize}
\item	Until $y=s(a'_p)=s(a'_{p'})$ holds, if $a_i$ and $a_j$ interact at $C_u\rightarrow C_{u+1}$, $a'_i$ and $a'_j$ interact at $C'_u\rightarrow C'_{u+1}$.
\item To define the remainder of $E'$, use $Proc(q,q')$ that is defined in Theorem 11 of \cite{bipartition}.
Procedure $Proc(q,q')$, which uses two indices $q$ and $q'$, can be applied to a configuration that satisfies $y=s(a'_p)=s(a'_{p'})=s(a'_{P-1})=s(a'_P)$ holds. $Proc(q,q')$ creates a sub-execution similar to $E$ by making $a'_{q}$ and $a'_{q'}$ join interactions instead of $a'_p$ and $a'_{p'}$, and guarantees that all agents in $A(q,q')=(A'- \{a'_p,a'_{p'}, a'_{P-1},a'_{P}\}) \cup \{ a'_q, a'_{q'} \}$ interact each other and the last configuration also satisfies the above condition. 
The formal definition of $Proc(q,q')$ is as follows. 
\begin{itemize}
\item When $a_i$ and $a_j$ interact at $C_u\rightarrow C_{u+1}$, $a'_i$ and $a'_j$ interact at $C'_u\rightarrow C'_{u+1}$ if $i$, $j \notin \{p,p'\}$. 
\item If $i=p$ or $j=p$ holds, $a'_q$ joins the interaction instead of $a'_p$.
\item If $i=p'$ or $j=p'$ holds, $a'_{q'}$ joins the interaction instead of $a'_{p'}$.
\item Procedure $Proc(q,q')$ continues these behaviors until all agents in $A(q,q')$ interact each other and satisfy $s(a'_q)=s(a'_{q'})=y$.
\end{itemize} 
	By using $Proc(q,q')$, we define the remainder of $E'$ to satisfy weak fairness as follows: Repeat $Proc(p,p')$, $Proc(p,P-1)$, $Proc(p,P)$, $Proc(P-1,p')$, $Proc(P,p')$, and $Proc(P,P-1)$.
	\end{itemize}
	Clearly, $E'$ makes $P-2$ agents behave similarly to $E$ and two agents keep state $y$. Hence, $E'$ never converges to a stable configuration.
	Since $E'$ is weakly-fair, this is a contradiction.
\end{proof}
\setcounter{lemma}{4}


\setcounter{lemma}{1}
\begin{lemma}
When the number of agents other than the BS is $P-2$, there exists an execution of $Alg$ that includes a strongly-stable configuration.
\end{lemma}

\begin{proof}
For contradiction, we assume that such an execution does not exist.

First, consider a weakly-fair execution $E=C_0, C_1, C_2, \ldots$ of $Alg$. By Lemma \ref{lem:P-2blue}, after some configuration $C_t$, all agents have different states. From the assumption, $C_t$ is not a strongly-stable configuration. That is, there exists a configuration reachable from $C_t$ such that two agents have the same state.

Hence, we can construct another weakly-fair execution $E'=C'_0, C'_1, C'_2, \ldots, C'_t, \ldots C'_{u}, \ldots$ of $Alg$ as follows:
\begin{itemize}
\item For $0 \le i \le t$, $C'_i$ is equal to $C_i$.
\item After $C'_t$, the population transits to $C'_{u}$ such that two agents have the same state.
\item After $C'_u$, agents continue to interact so that $E'$ satisfies weak fairness.
\end{itemize}
 By Lemma \ref{lem:P-2blue}, since $E'$ is weakly-fair, all agents have different states after some configuration $C'_{t'}$.

Hence, similarly to $E'$, we can construct another weakly-fair execution $E''=C''_0$, $C''_1$, $\ldots$, $C''_{t'}$, $\ldots$, $C''_{u'}$, $\ldots$ that satisfies the following conditions:
\begin{itemize}
\item For $0 \le i \le t'$, $C''_i$ is equal to $C'_i$.
\item After $C''_{t'}$, the population transits to $C''_{u'}$ such that two agents have the same state. 
\item After $C''_u$, agents continue to interact so that $E''$ satisfies weak fairness.
\end{itemize}

By repeating this construction, we can construct a weakly-fair execution such that two agents have the same state infinitely often. 
From Lemma \ref{lem:P-2blue}, this is a contradiction.
\end{proof}

\setcounter{lemma}{0}

\section{Proofs of Lemmas \ref{lem:P-2blue2} and \ref{lem:Q^+}}
\label{app:lem34}
From now, we show the proofs of Lemmas \ref{lem:P-2blue2} and \ref{lem:Q^+}.
First, we show a proof of Lemma \ref{lem:P-2blue2}.

\setcounter{lemma}{2}
\begin{lemma}
For any distinct states $p$ and $q$ such that $p\in Q_{blue}$ and $q \in Q_{t}$ hold, transition rule $(p,q) \rightarrow (p',q')$ satisfies the following conditions.

\begin{itemize}
\item If $q\in Q_{red}$ or $q\in Q_b$ (i.e., $q$ is a state of the BS) holds, $p' = p$ holds.
\item If $q\in Q_{blue}$ holds, either $(p', q')=(p, q)$ or $(p', q')=(q, p)$ holds.
\end{itemize}
\end{lemma}

\begin{proof}
Consider execution $E_\alpha$ defined in Definition \ref{def:E}. From the definition, every configuration $C_u$ ($u\ge t$) is strongly-stable. 
From Corollary \ref{cor:Sblue}, each $blue$ state exists in any strongly-stable configuration.
For this reason, $p$ exists in any configuration $C_u$ ($u\ge t$).
In addition, since $q \in Q_{t}$ holds, there exists a configuration $C_v$ ($v \ge t$) such that state $q$ exists at $C_v$.
Let $a_x$ and $a_y$ be agents that have states $p$ and $q$ at $C_v$, respectively. We consider another execution $E'_\alpha$ such that $a_x$ and $a_y$ interact at configuration $C_v$, that is, they change their states by transition rule $(p,q) \rightarrow (p',q')$.
Let $C'_v$ be the resultant configuration.
Note that the transition does not change colors of $a_x$ and $a_y$ because $C_v$ is stable.
Now, we consider the following two cases.

First, consider the case of $q \in Q_{red}$ or $q \in Q_b$.
For contradiction, assume that $p' \neq p$ holds. Since $p' \in Q_{blue}$ holds, $p'$ exists in $C_v$ by Corollary \ref{cor:Sblue} (if $P=2$ holds, such $p'$ may not exist in $C_v$. However, if $P=2$ holds, clearly $Q_{blue} = \{p\}$ holds. Hence, since $p' \in Q_{blue}$ holds, $p' = p$ holds and then the lemma holds immediately in this case).
Thus, in $C'_v$, two agents have state $p'$.
This is a contradiction because $C_v$ is strongly-stable.

Next, consider the case of $q \in Q_{blue}$.
For contradiction, assume that 1) $(p', q')=(p, p)$, 2)$(p', q')=(q, q)$, or 3)$p'=r \vee q' = r$ for some $r \notin \{p,q\}$ holds.
In the former two cases, two agents have the same state in $C'_v$. In the last case, since $r \in Q_{blue}$ holds, state $r$ exists in $C_v$ from Corollary \ref{cor:Sblue} and consequently two agents have $r$ in $C'_v$ (if $P=2$ holds, such $r$ may not exist in $C_v$. However, if $P=2$ holds, there does not exist $r$ such that $r \notin \{p,q\}$ holds).
This is a contradiction because $C_v$ is strongly-stable.
Therefore, the lemma holds.
\end{proof}
\setcounter{lemma}{4}

In the following, we prove Lemma \ref{lem:Q^+}. 

By Corollary \ref{cor:P-2-2}, without loss of generality, we assume that $Q_{blue}=\{s_1,s_2,\ldots,s_{P/2-1}\}$ and $Q_{red}=\{s_{P/2},s_{P/2+1},\ldots,s_{P-1}\}$ hold.

\begin{definition}
\label{def:Qbluec}
$Q_{bc}=\{p \in Q_{blue} \mid \exists q \in Q_{red}:p \rightsquigarrow q\}$.
\end{definition}

\begin{definition}
\label{def:Qnbc}
$Q_{nbc}=Q_{blue} \backslash Q_{bc}$.
\end{definition}

Intuitively, $p \in Q_{bc}$ means that a $blue$ agent with state $p$ can become $red$ by interactions with homonyms.
Also, $p \in Q_{nbc}$ means that a $blue$ agent with state $p$ cannot become $red$ by interactions with homonyms.

The outline of proof of Lemma \ref{lem:Q^+} is as follows. 
First, we show that $Q_{nbc}$ is not empty. 
This implies that, by the definition of $Q_{nbc}$, there exists a state $p$ in $Q_{nbc}$ such that $p \rightsquigarrow p$ holds.
After that, we show that there exists a state $p^*$ in $Q_{nbc}$ such that $p^* \overset{*}{\rightsquigarrow} p^*$ holds.
If such $p^*$ does not exist, some state in $Q_{nbc}$ can transit to a state in $Q_p \backslash Q_{nbc}$ by interactions with homonyms and that contradicts the definition of $Q_{nbc}$.
The existence of $p^*$ implies that there exists a state set $Q_{p*} \subseteq Q_{nbc}$ such that $Q_{p^*}=\{ q \mid p^* \overset{*}{\rightsquigarrow} q \}$ holds.
Finally, we show that this $Q_{p*}$ satisfies the condition of Lemma \ref{lem:Q^+}.

\begin{definition}
\label{def:tradis}
$\DtR{s_i}$ is a function that satisfies the following property.
\begin{itemize}
\item If $s_{i} \in Q_{red}$ holds, $\DtR{s_i}=0$ holds.
\item If $s_i \in Q_{bc}$ holds, $\DtR{s_i}= \min\{\DtR{s^1_j}, \DtR{s^2_j}\}+1$ holds when transition rule $(s_i,s_i)\rightarrow (s^1_j,s^2_j)$ exists.
\item If $s_i \in Q_{nbc}$ holds, $\DtR{s_i}= \infty$ holds.
\end{itemize}

\end{definition}

Intuitively, $DtR(s_i)$ gives the minimum number of interactions to transit from $s_i$ to a $red$ state when allowing only interactions with homonyms. 
The following lemma shows that $Q_{nbc}$ is not empty.

\begin{lemma}
\label{lem:Qnbc}
$Q_{nbc} \neq \emptyset$. 
\end{lemma}
\begin{proof}
The idea of proof of the lemma is as follows.
For contradiction, we assume that $Q_{nbc} \neq \emptyset$ does not hold.
When $P$ is even, some $blue$ agents have the same state in a stable configuration because the number of $blue$ agents should be $P/2$ and $|Q_{blue}|=P/2-1$ holds.
This implies that interactions with homonyms are always possible. We show that this makes some $blue$ agent transit to $red$.

From now on, we show the details of proof.
For contradiction, assume $Q_{nbc} = \emptyset$.

Consider a population $A'=\{a'_0, \ldots, a'_{P}\}$ of $P$ agents and an initialized BS, where $a'_0$ is the BS.
Let $E'=C'_0, C'_1, \ldots, C'_t, \ldots$ is a weakly-fair execution of $Alg_{asym}$, where $C'_t$ is a stable configuration.

W.l.o.g., assume $\DtR{s_{P-1}} = \cdots = \DtR{s_{P/2}} = 0 < \DtR{s_{P/2-1}} \le \cdots \le \DtR{s_{1}}$.
Since $Q_{nbc} = \emptyset$ holds, $\DtR{s_1} \neq \infty$ holds. From the definition of $DtR(s_i)$, if $DtR(s_i)>0$ and $DtR(s_i) \neq \infty$ hold, an agent with state $s_i$ transits to $s_j$ with $j>i$ by an interaction of homonyms.
By using this property, we prove the following lemma.
\begin{lemma}
\label{lem:redtra}
For $1 \le i \le P/2-1$, if $i+1$ agents have states in $\{s_1,\ldots,s_i\}$, one of these agents can transit to $s_{h}(h \ge i+1)$ by interactions among these agents.
\end{lemma}
\begin{proof}
We show the lemma by induction of $i$. 

The base case is when two agents have state $s_1$. If they interact, one of them transits to $s_{h}(h \ge 2)$.
Thus, the lemma holds in this case.

For the induction step, we assume that the lemma holds for $i = k-1$ ($k \le P/2-1$), and prove that the lemma holds for $i=k$. To do this, consider the situation that $k+1$ agents have states in $\{s_1$, $\ldots$, $s_k\}$.

Consider three cases (1) at least two agents have state $s_k$, (2) exactly one agent has state $s_k$, and (3)no agent has state $s_k$.

First, consider the case that at least two agents have state $s_k$.
In this case, if these two agents interact, one of them transits to $s_{h}(h \ge k+1)$.
Hence, the lemma holds in this case.

Next, consider the case that exactly one agent has state $s_k$.
In this case $k$ agents have states in $\{s_1$, $\ldots$, $s_{k-1}\}$.
Hence, from the inductive assumption, one of them can transit to $s_{h}(h \ge k)$ by interactions among the $k$ agents.
If $h>k$ holds, the lemma holds.
Otherwise, two agents have state $s_k$ and hence the lemma holds similarly to the first case.

Finally, consider the case that no agent has state $s_k$.
This implies that $k+1$ agents have states in $\{s_1$, $\ldots$, $s_{k-1}\}$.
Hence, from the inductive assumption, two of them can transit to $s_{h}(h \ge k)$ and $s_{h'}(h' \ge k)$.
If $h>k$ or $h'>k$ holds, the lemma holds.
Otherwise, two agents have state $s_k$ and hence the lemma holds similarly to the first case.
\end{proof}

Since $C'_t$ is a stable configuration in $E'$, $P/2$ agents have states in $Q_{blue}= \{s_1$, $\ldots$, $s_{P/2-1}\}$.
Hence, by Lemma \ref{lem:redtra}, we can construct an execution segment that makes a $blue$ agent transit to $s_h$ for some $h \ge P/2$. This implies that the agent changes its color from $blue$ to $red$. Since $C'_t$ is a stable configuration, this is a contradiction.
\end{proof}

From now, we show that there exists a state $p \in Q_{nbc}$ such that $p \rightsquigarrow p$ holds.
Furthermore, we also show that there exists a state $p \in Q_{nbc}$ such that $p \overset{*}{\rightsquigarrow} p$ holds.

\begin{lemma}
\label{lem:p*p}
There exists a state $p \in Q_{nbc}$ such that $p \rightsquigarrow p$ holds.
\end{lemma}
\begin{proof}
From Lemma \ref{lem:Qnbc}, there exists a state $p_0$ in $Q_{nbc}$.
From the definition, there exists a sequence of transition rules $(p_0,p_0) \rightarrow (p_1,x_0)$, $(p_1,p_1) \rightarrow (p_2,x_1)$, $(p_2,p_2) \rightarrow (p_3,x_2)$, $\ldots$ starting from $p_0$.
Since the number of states is finite, there exist some $p_i$ and $p_j$ such that $j > i \geq 0$ and $p_i=p_j$ holds. This implies $p_i \rightsquigarrow p_i$. 
Clearly, $p_i \in Q_{nbc}$ holds because $p_0 \in Q_{nbc}$ holds.
Therefore, the lemma holds.
\end{proof}

\begin{lemma}
\label{lem:p*p2}
There exists a state $p \in Q_{nbc}$ such that $p \overset{*}{\rightsquigarrow} p$ holds.
\end{lemma}
\begin{proof}
For contradiction, assume that such state $p$ does not exist.
By Lemma \ref{lem:p*p}, there exists a state $p_0 \in Q_{nbc}$ such that $p_0 \rightsquigarrow p_0$ holds.

Since $p_0 \overset{*}{\rightsquigarrow} p_0$ does not hold, there exists some state $q$ such that $p_0 \rightsquigarrow q$ holds but $q \rightsquigarrow p_0$ does not hold.
Since $p_0 \rightsquigarrow q$ holds, $q$ belongs to $Q_{nbc}$.
For this reason, there exists a state $p'_0$ such that $q \rightsquigarrow p'_0 \rightsquigarrow p'_0$ holds.
Note that, since $q \rightsquigarrow p_0$ does not hold, $p'_0$ is not equal to $p_0$.

Also, since $p'_0 \overset{*}{\rightsquigarrow} p'_0$ does not hold, there exists some state $q'$ such that $p'_0 \rightsquigarrow q'$ holds but $q' \rightsquigarrow p'_0$ does not hold.
This implies that there exists a state $p''_0 \in Q_{nbc}$ such that $q' \rightsquigarrow p''_0 \rightsquigarrow p''_0$, $p''_0 \neq p'_0$, and $p''_0 \neq p_0$ hold.

By repeating similar arguments, since the number of states is finite, we can prove that there exists some state $p^*_0 \in Q_{nbc}$ such that $p^* \overset{*}{\rightsquigarrow} p^*$ holds.
This is a contradiction.
\end{proof}

In the following, we focus on a state $p^* \in Q_{nbc}$ such that $p^* \overset{*}{\rightsquigarrow} p^*$ holds.
Let $Q_{p^*}=\{q \mid p^* \rightsquigarrow q\}$.
Note that, if homonyms with a state in $Q_{p*}$ interact, they transit to a state in $Q_{p*}$. 
This implies that $Q_{p^*}$ satisfies the first condition of $Q^*$ in Lemma \ref{lem:Q^+}.
In the following lemma, we prove that $Q_{p*}$ satisfies the second condition of $Q^*$ in Lemma \ref{lem:Q^+}.

\begin{lemma}
\label{lem:Ttra}
Consider a population $A_{p*}=\{a_1$, $a_2$, $\ldots$, $a_z\}$, where $z=|Q_{p^*}|$ holds.
Consider an initial configuration $C^q_0$ such that all agents in $A_{p*}$ have states in $Q_{p*}$. For any $q \in Q_{p*}$, there exists an execution segment $e^q=C^q_0$, $C^q_1$, $C^q_2$, $\ldots$, $C^q_m$ such that 1) some agent has state $q$ at the last configuration $C^q_m$ and 2) only homonyms interact in $e^q$.
\end{lemma}

\begin{proof}
Let $C$ be a configuration with $A_{p*}$ such that all agents have states in $Q_{p*}$ and there does not exist an agent with $s$ in $C$.
Since $|Q_{p^*}|$ agents have states in $Q_{p*}$ in $C$ and there does not exist an agent with $s$ in $C$, there exist homonyms in $C$.
When homonyms with a state in $Q_{p*}$ interact, they transit to a state in $Q_{p*}$.
These imply that, when homonyms interact at $C \rightarrow C'$, either an agent with $s$ or homonyms with a state in $Q_{p*}$ exist in $C'$.
Thus, for contradiction, assume that there exists an infinite execution segment $e^q=C^q_0$, $C^q_1$, $C^q_2$, $\ldots$ with $A_{p*}$ such that only homonyms interact and any agent never has $q$ in $e^q$.
Consider the case that $a_x$ and $a_y$ interact at $C^q_i \rightarrow C^q_{i+1}$ for $i \ge 0$.
By the assumption, $a_x$ and $a_y$ have the same state $p \in Q_{p*}$.
From the property of $Q_{p*}$, $p$ satisfies $p \rightsquigarrow p^* \rightsquigarrow q$.
Thus, $\DtQ{s(a_x,C^q_i)}{q} > \DtQ{s(a_x,C^q_{i+1})}{q} $ or $\DtQ{s(a_y,C^q_i)}{q} > \DtQ{s(a_y,C^q_{i+1})}{q} $ holds.
Hence, from the property of $\Phi$, $\Phi(C^q_i,q) > \Phi(C^q_{i+1},q)$ holds.
Since the possible value of $\Phi(C,q)$ is finite, $\Phi(C^q_j,q)$ includes $0$ for some $C^q_j$ and thus some agent has $q$ in $C^q_j$.
By the definition of $e^q$, this is a contradiction.
\end{proof}

\setcounter{lemma}{3}
\begin{lemma}
There is a non-empty state set $Q^* \subseteq Q_{blue}$ that satisfies the following conditions.

\begin{itemize}
\item For any state $p \in Q^*$, transition rule $(p,p) \rightarrow (p',q')$ satisfies $p'\in Q^*$ and $q'\in Q^*$.
\item Assume that, in a configuration $C$, there exists a subset of agents $A^*$ such that all agents in $A^*$ have states in $Q^*$ and $|A^*|=|Q^*|+1$ holds. In this case, for any agent $a_r \in A^*$ and any state $q \in Q^*$, there exists an execution segment such that 1) the execution segment starts from $C$, 2) $a_r$ has state $q$ at the last configuration, 3) only agents in $A^*$ join interactions, and 4) all agents in $A^*$ have states in $Q^*$ at the last configuration.
\end{itemize}
\end{lemma}

\begin{proof}
We show that $Q_{p^*}$ satisfies the condition of $Q^*$. Clearly $Q_{p^*}$ satisfies the first condition. Hence, we focus on the second condition.

Consider a set of agents $A^*$, and consider an initial configuration $C^{p^*}_0$ such that all agents in $A_{p*}$ have states in $Q_{p*}$.
Let $a_r$ be an agent in $A^*$ and let $s = s(a_r, C^{p^*}_0)$.
Consider a sequence of states $T=t_0, t_1, t_2, \ldots, t_l$ such that $t_0=s$, $t_l=q$, and, for any $i$ ($0 \le i < l $), transition rule $(t_i,t_i) \rightarrow (t_{i+1},x_i)$ or $(t_i,t_i)\rightarrow(x_i,t_{i+1})$ exists for some $x_i$.

From configuration $C^{p^*}_0$, we make $a_r$ change its state according to $T$.
That is, if $a_r$ has state $t_i(0 \le i < l)$, we make one of the remaining $|Q_{p^*}|$ agents (i.e., $A^*-\{a_r\}$) transit to $t_i$ by interactions of homonyms and then make the agent interact with $a_r$.
Such procedure is possible because, by Lemma \ref{lem:Ttra}, one of the remaining agents can transit to any state in $Q_{p^*}$ by interactions of homonyms. Note that all agents in $A^*$ keep states in $Q_{p^*}$ when the procedure is applied. Hence, we can construct an execution segment in the second condition.
Therefore, $Q_{p^*}$ satisfies the conditions of $Q^*$ and thus the lemma holds.
\end{proof}
\setcounter{lemma}{9}

\setcounter{theorem}{1}

\section{Proof of Theorem \ref{the:sym}}
\label{app:imp2}

\begin{theorem}
In the model with an initialized BS, there is no symmetric protocol that solves the uniform bipartition problem with $P$ states from arbitrary initial states under weak fairness, if $P$ is an even integer.
\end{theorem}

Assume, by way of contradiction, that such algorithm $Alg_{sym}$ exists.
Let $Q_p=\{s_1, \ldots, s_{P}\}$ be a set of states of non-BS agents.
Let $Q_{blue}=\{ s \in Q_p \mid f(s)=blue\}$ be a set of blue states and $Q_{red}=\{ s \in Q_p \mid f(s)=red\}$ be a set of red states.
Without loss of generality, we assume that $|Q_{blue}| \le |Q_{red}|$ holds.

First, we introduce a sink state that is defined in \cite{burman2018brief}.

\begin{definition}
\label{def:reatra}
For states $q$ and $q'$, we say $q \overset{sym}{\rightsquigarrow} q'$ if there exists a sequence of states $q = q_0, q_1, \cdots , q_k = q'$ such that, for any $i(0 \le i < k)$, 
\end{definition}

\begin{definition}
\label{def:reatra}
For state $q$, if $q \overset{sym}{\rightsquigarrow} q$ holds, $q$ is called a loop state.
\end{definition}

\begin{definition}
\label{def:sinks}
State $m$ is a sink state if $m\in Q_p$ satisfies the following conditions:
\begin{enumerate}
\item There exists a transition rule $(m, m) \rightarrow (m, m)$.
\item For any $s \in Q_p$, $s \overset{sym}{\rightsquigarrow} m$ holds.
\item If the number of agents is at most $P-1$, $m$ does not occur infinitely often for any execution.
\end{enumerate}
\end{definition}

In the case of naming protocols \cite{burman2018brief}, the impossibility proof proves that a sink state always exists. However, in the case of uniform bipartition protocols, sometimes no sink state exists. To treat this situation, we additionally define a sink pair, which is a pair of two states that has a similar property of a sink state. We will prove that either a sink state or a sink pair exists.
\begin{definition}
\label{def:sinkp}
A pair of two states $m_1, m_2\in Q_p$ is a sink pair if the following conditions hold:
\begin{enumerate}
\item There exist transition rules (1) $(m_1,m_1) \rightarrow (m_1,m_1)$ and $(m_2,m_2) \rightarrow (m_2,m_2)$ or (2) $(m_1,m_1) \rightarrow (m_2,m_2)$ and $(m_2,m_2) \rightarrow (m_1,m_1)$.
\item For any $s \in Q$, $s \overset{sym}{\rightsquigarrow} m_1$ or $s \overset{sym}{\rightsquigarrow} m_2$ holds.
\item If the number of agents is at most $P-2$, $m_1$ and $m_2$ do not occur infinitely often for any execution.
\end{enumerate}
\end{definition}

Note that, if a sink state exists, a sink pair does not exist, and vice versa.
The following lemma gives an important property to prove the existence of a sink state or a sink pair.

\begin{lemma}
\label{lem:sink}
Let $E = C_0, C_1, C_2, \ldots$ be a weakly-fair execution of $Alg_{sym}$ with $n \le P-2$ agents and an initialized BS. For any loop state $s_r \in Q_p$ (that is, $s_r \overset{sym}{\rightsquigarrow} s_r$ holds), $s_r$ does not occur infinitely often in $E$.
\end{lemma}

\begin{proof}
The idea of the proof is as follows.
First, for contradiction, we consider an execution $E$ with $n \le P-2$ such that $s_r$ occurs infinitely often in $E$.
Next, we consider an execution with $P$ agents such that all additional agents have $s_r$ as their initial states and other agents behave similarly to $E$.
In the execution, all additional agents do not join the interactions until some other agent has $s_r$. 
At that time, one of non-additional agents has state $s_r$ and additional agents also have state $s_r$. 
We can prove that, from this configuration, non-additional agents cannot recognize the additional agents and hence they make the same behavior as in $E$.
In addition, the additional agents can keep state $s_r$.
Since the numbers of $red$ and $blue$ agents are balanced without the additional agents and the additional agents have the same state, the uniform bipartition problem cannot be solved.

From now on, we show the details of proof. 
For contradiction, we assume that there exists a weakly-fair execution $E$ such that $s_r$ occurs infinitely often in $E$.
First, consider a population $A = \{ a_0, a_1, a_2, \ldots, a_n \}$ of $n\le P-2$ agents and an initialized BS, where $a_0$ is the BS.
Since there exist a finite number of agents, there exists a particular agent $a_x$ that has $s_r$ infinitely often in $E$.
We can define infinite configurations $C_{u_0}, C_{u_1}, \ldots$ and infinite execution segments $e_0, e_1, e_2, \ldots$ of $E$ so that $E = C_{u_0}$, $e_0$, $C_{u_1}$, $e_1$, $C_{u_2}$, $e_2$, $C_{u_3}$, $\ldots$ satisfies the following:
\begin{itemize}
\item For $w \ge 1$, agent $a_x$ has state $s_r$ in $C_{u_w}$.
\item For $j \ge 0$, during execution segment $C_{u_{j}},e_j,C_{u_{j+1}}$, any pair of agents in $A$ interact at least once (this is possible because $E$ is weakly-fair).
\end{itemize}

Next, consider a population $A' = \{a'_0, a'_1, \ldots$, $a'_P\}$ of $P$ agents and an initialized BS, where $a'_0$ is the BS.
We define execution $E'= C'_{u_0}$, $e'_0$, $C'_{v_0}$, $e^m_0$, $C'_{u_1}$, $e'_1$, $C'_{v_1}$, $e^m_1$, $C'_{u_2}$, $e'_2$, $C'_{v_2}$, $\ldots$ as follows:
\begin{itemize}
\item In initial configuration $C'_{u_0}$, $a'_i$ ($n \ge i \ge 0$) has the same state as $a_i$ in $C_0$ and $a'_{n+1},a'_{n+2},\ldots,a'_P$ have state $s_r$. Formally, $s(a'_i, C'_{u_0}) = s(a_i, C_{u_0})$ holds for any $i$ ($n\ge i \ge 0$), and $s(a'_{n+1}, C'_{u_0}) = s(a'_{n+2}, C'_{u_0}) = \cdots = s(a'_{P}, C'_{u_0}) = s_r$ holds.
\item For $j \ge 0$, we construct execution segment $e'_j = \Ch{j}{1}$, $\Ch{j}{2}$, $\Ch{j}{3}$, $\ldots$, $\Ch{j}{z}$ and configuration $C'_{v_j}$ by using $e_j = \C{j}{1}$, $\C{j}{2}$, $\C{j}{3}$, $\ldots$, $\C{j}{z}$ and $C_{u_{j+1}}$, where $z=|e_j|$ holds.
Concretely, we construct $e'_j$ as follows:
\begin{itemize}
\item Case that $j=(P-n+1) \cdot y$ holds for some $y$ ($y\ge 0$). In this case, agents $a'_0,\ldots,a'_n$ interact in execution segment $C'_{u_j},e'_j,C'_{v_{j}}$ similarly to $a_0,\ldots,a_n$ in execution segment $C_{u_j},e_j,C_{u_{j+1}}$. Formally, when $a_g$ and $a_h$ interact at $\C{j}{f} \rightarrow \C{j}{f+1}$ for $z>f>0$ (resp., $C_{u_{j}} \rightarrow \C{j}{1}$ and $\C{j}{z} \rightarrow C_{u_{j+1}}$), $a'_g$ and $a'_h$ interact at $\Ch{j}{f} \rightarrow \Ch{j}{f+1}$ (resp., $C'_{u_{j}} \rightarrow \Ch{j}{1}$ and $\Ch{j}{z} \rightarrow C'_{v_{j}}$).
\item Case that $j=(P-n+1) \cdot y + l$ holds for some $y$ ($y\ge 0$) and $l$ ($P-n \ge l \ge 1$). In this case, $a'_{n+l}$ joins interactions instead of $a'_x$. Note that, in $C'_{u_j}$, both $a'_{n+l}$ and $a'_x$ have state $s_r$. Formally we construct $e'_j$ as follows: 
(1) when $a_g (g \neq x)$ and $a_h (h\neq x)$ interact at $\C{j}{f} \rightarrow \C{j}{f+1}$ for $z>f>0$ (resp., $C_{u_{j}} \rightarrow \C{j}{1}$ and $\C{j}{z} \rightarrow C_{u_{j+1}}$), $a'_g$ and $a'_h$ interact at $\Ch{j}{f} \rightarrow \Ch{j}{f+1}$ (resp., $C'_{u_{j}} \rightarrow \Ch{j}{1}$ and $\Ch{j}{z} \rightarrow C'_{v_{j}}$), 
(2) when $a_x$ interacts with an agent $a_i(i \neq x)$ at $\C{j}{f} \rightarrow \C{j}{f+1}$ for $z>f>0$ (resp., $C_{u_{j}} \rightarrow \C{j}{1}$ and $\C{j}{z} \rightarrow C_{u_{j+1}}$), $a'_{n+l}$ interacts with $a'_i$ at $\Ch{j}{f} \rightarrow \Ch{j}{f+1}$ (resp., $C'_{u_{j}} \rightarrow \Ch{j}{1}$ and $\Ch{j}{z} \rightarrow C'_{v_{j}}$).
\end{itemize}
\item For $j \ge 0$, during execution segment $C'_{v_j},e^m_j,C'_{u_{j+1}}$, agents $a'_x$, $a'_{n+1}$, $a'_{n+2}$, $\ldots$, $a'_P$ interact so that any pair of them interact at least once and eventually they have state $s_r$.
At the first configuration, agents $a'_x$, $a'_{n+1}$, $a'_{n+2}$, $\ldots$, $a'_P$ have state $s_r$.
Since $s_r \overset{sym}{\rightsquigarrow} s_r$ holds, each pair of them can go back to state $s_r$ after some interactions.
Thus, in $C'_{u_j}(j>0)$, agents $a'_x$, $a'_{n+1}$, $a'_{n+2}$, $\ldots$, $a'_P$ have state $s_r$.
\end{itemize}

For $j=(P-n+1) \cdot y$, in execution segment $C'_{u_{j}},e'_j,C'_{v_j}$, agents in $A'- \{a'_{n+1}, a'_{n+2}, \ldots, a'_{P}\}$ interact each other.
For $j=(P-n+1) \cdot y + l$, in execution segment $C'_{u_{j}},e'_j,C'_{v_j}$, agents in $(A'- \{a'_{x}, a'_{n+1}, a'_{n+2}, \ldots, a'_{P}\}) \cup \{ a'_{n+l} \}$ interact each other.
Moreover, $a'_x$, $a'_{n+1}$, $a'_{n+2}$, $\ldots$, $a'_P$ interact each other in $e^m_j$ for $j>0$.
From these facts, execution $E'$ is weakly-fair.

In $E$, let $C_{u_t}$ be a stable configuration such that agent $a_x$ has state $s_r$. 
Let $R_{u_t}$ be a set of $red$ agents in $C_{u_t}$ and let $B_{u_t}$ be a set of $blue$ agents in $C_{u_t}$. Clearly, $||R_{u_t}|-|B_{u_t}||\le 1$ holds.
Now, we consider two cases. One is the case that $n$ is even (i.e., $||R_{u_t}|-|B_{u_t}||=0$ holds). Another is the case that $n$ is odd (i.e., $||R_{u_t}|-|B_{u_t}||=1$ holds).
Note that, in both cases, $s(a_i, C_{u_w}) = s(a'_i, C'_{u_w})$ holds for $n \ge i \ge 0$ and $w \ge 0$, and other $P-n$ agents have state $s_r$ in $C'_{u_w}$ for $w \ge 0$.
Hence, for the number of $f(s_r)$-agents, the difference between $C_{u_w}$ and $C'_{u_w}$ is $P-n$ for any $w \ge t$.

First, we consider the case that $n$ is even.
After $C'_{u_t}$, the number of $f(s_r)$-agents is $P-n \ge 2$ more than the number of $\overline{f(s_r)}$-agents.
Thus, $E'$ never reaches a stable configuration.

Next, we consider the case that $n$ is odd.
Since $n<P-1$ holds and both $P-1$ and $n$ are odd, $n \le P-3$ holds and thus $P-n \ge 3$ holds.
Hence after $C'_{u_t}$, the number of $f(s_r)$-agents is at least two more than the number of $\overline{f(s_r)}$-agents.
Thus, $E'$ never reaches a stable configuration.

Since $E'$ is weakly-fair, this is a contradiction.
\end{proof}

Using Lemma \ref{lem:sink}, Lemma \ref{pro:sink} shows the existence of a sink state or a sink pair.

\begin{lemma}
\label{pro:sink}
In any protocol $Alg_{sym}$, there exists either exactly one sink pair or exactly one sink state.
\end{lemma}

\begin{proof}
For $q \in Q_p$, let $L_{q}=\{ q' \mid q \overset{sym}{\rightsquigarrow} q'$ and $q' \overset{sym}{\rightsquigarrow} q'\}$. That is, $L_q$ is a set of loop states such that an agent with state $q$ can transit to the state by interactions of homonyms. For any $q_0\in Q_p$, we can consider a sequence of transition rules $(q_0,q_0) \rightarrow (q_1,q_1), (q_1,q_1) \rightarrow (q_2,q_2) \rightarrow \cdots$. Because the number of states is finite, $q_i=q_j$ holds for some $j>i \ge 1$. Hence, $L_q\neq \emptyset$ holds. 
We define $L$ as $L =  L_{s_1} \cup L_{s_2} \cup L_{s_3} \cup \cdots \cup L_{s_P}$.

From now, we show that $|L| \le 2$ holds by Lemmas \ref{lem:sink} and \ref{lem:strcon}.
Recall that the properties of Lemma \ref{lem:strcon} hold even in a symmetric algorithm $Alg_{sym}$.
By Lemma \ref{lem:sink}, a loop state does not occur infinitely often in any execution with $n < P-1$ agents.
In addition, by Lemma \ref{lem:strcon}, when the number of agents is $P-2$, there exists an execution such that at least $P-2$ states occur infinitely often.
This implies that such $P-2$ states are not loop states.
Thus, the number of loop states is at most two, that is, $|L| \le 2$ holds.

If $|L|=1$ holds, a loop state in $L$ satisfies the conditions 2 and 3 of a sink state in Definition \ref{def:sinks}.
This is because an agent with any state can transit to the loop state in $L$ by interactions with homonyms and the loop state does not occur infinitely often by Lemma \ref{lem:sink}.
For a similar reason, if $|L|=2$ holds, two loop states in $L$ satisfy the conditions 2 and 3 of a sink pair in Definition \ref{def:sinkp}.
Note that, in this case, the two loop states in $L$ are not sink states because they cannot satisfy conditions 1 and 2 of a sink state in Definition \ref{def:sinks} at the same time.

In the following, we show that a loop state (resp., two loop states) in $L$ satisfies the condition 1 of a sink state (resp., a sink pair).

First, we consider the case of $|L|=2$.
Let $m_1$ and $m_2$ be states in $L$.
For contradiction, we assume that there exists $(m_1, m_1) \rightarrow (s, s)$ such that $s \notin \{m_1, m_2\}$ holds.
Since $m_1  \overset{sym}{\rightsquigarrow}  m_1$ holds, $s  \overset{sym}{\rightsquigarrow} s$ holds.
However, by the assumption, there does not exist such $s$ because only $m_1$ and $m_2$ are loop states.
Hence, $(m_1, m_1) \rightarrow (s, s)$ does not exist, and hence either $(m_1, m_1) \rightarrow (m_1, m_1)$ or $(m_1, m_1) \rightarrow (m_2, m_2)$ exists. 
Similarly, $(m_2,m_2) \rightarrow (s,s)$ for $s \notin \{m_1,m_2\}$ does not exist, and hence either $(m_2,m_2) \rightarrow (m_1,m_1)$ or $(m_2,m_2) \rightarrow (m_2,m_2)$ exists.
If both $(m_1, m_1) \rightarrow (m_1, m_1)$ and $(m_2, m_2) \rightarrow (m_1, m_1)$ exist, $m_2  \overset{sym}{\rightsquigarrow}  m_2$ does not hold. 
Similarly, if both $(m_1, m_1) \rightarrow (m_2, m_2)$ and $(m_2, m_2) \rightarrow (m_2, m_2)$ exist, $m_1  \overset{sym}{\rightsquigarrow}  m_1$ does not hold. 
Thus, $m_1$ and $m_2$ satisfy the condition 1 of a sink pair.

Next, we consider the case of $|L|=1$.
Let $m$ be a state in $L$.
For contradiction, we assume that there exists $(m, m) \rightarrow (s, s)$ such that $s \neq m$ holds.
Since $m  \overset{sym}{\rightsquigarrow}  m$ holds, $s  \overset{sym}{\rightsquigarrow} s$ holds.
However, the loop state is only $m$.
This is a contradiction, and $m$ satisfies the condition 1 of a sink state.

Therefore, the lemma holds.
\end{proof}

We introduce a reduced execution that is also defined in \cite{burman2018brief}.
\begin{definition}
In a reduced execution, if homonyms with a non-sink state (resp., neither of the sink pair) occur, they are immediately transited to a sink state (resp., one of the sink pair) by interactions with the homonyms.
This procedure is called reducing. 
\end{definition}
By the condition 2 of the sink state and the sink pair, such reducing is possible.
We say configuration $C$ is reduced if there are no homonyms except agents with a sink state or one of the sink pair.
Note that there exists a reduced weakly-fair execution of $Alg_{sym}$ because any pair of agents can interact in a reduced configuration.

Consider a reduced weakly-fair execution $E$ of $Alg_{sym}$ with $P-2$ agents.
By Lemma \ref{lem:sink}, there exists a stable reduced configuration such that no agent has a sink state or states of the sink pair. Since no two agents have the same non-sink state, we have the following corollaries.

\begin{corollary}
\label{cor:P-22}
When a sink state exists in $Q_p$, either, 
\begin{itemize}
\item the number of non-sink $red$ states is $P/2$ and the number of non-sink $blue$ states is $P/2-1$, or, 
\item the number of non-sink $blue$ states is $P/2$ and the number of non-sink $red$ states is $P/2-1$.
\end{itemize}
\end{corollary}

\begin{corollary}
\label{cor:P-23}
When a sink pair exists in $Q_p$, the number of $red$ (resp., $blue$) states not in the sink pair is $P/2-1$ (resp., $P/2-1$).
\end{corollary}

Moreover, Corollary \ref{cor:P-22} can be extended to the following lemma.
\begin{lemma}
\label{lem:P-24}
When a sink state $m$ exists in $Q_p$, the number of non-sink $f(m)$-states is $P/2-1$ and the number of non-sink $\overline{f(m)}$-states is $P/2$.
This implies that the number of $f(m)$-states is $P/2$ and the number of $\overline{f(m)}$-states is also $P/2$.
\end{lemma}
\begin{proof}
By corollary \ref{cor:P-22}, the number of non-sink $f(m)$-states is at least $P/2-1$ and at most $P/2$.
Hence, for contradiction, we assume that the number of non-sink $f(m)$-states is $P/2$.
This implies that the number of $f(m)$-states (including the sink state $m$) is $P/2+1$ and the number of $\overline{f(m)}$-states is $P/2-1$.

Now, we consider a reduced weakly-fair execution $E$ of $Alg_{sym}$ with $P$ agents and an initialized BS.
In $E$, after a stable configuration, a reduced configuration occurs infinitely often.
In a reduced configuration, each non-sink state is held by at most one agent.
Thus, since all of $\overline{f(m)}$-states are non-sink states and the number of them is $P/2-1$, in a stable reduced configuration the number of $\overline{f(m)}$-agents is at most $P/2-1$ and the number of $f(m)$-agents is at least $P/2+1$.
This is a contradiction.
\end{proof}

Subsequently, we show that, when a sink state exists and the number of agents is $P-1$, the number of $f(m)$-agents is less than the number of $\overline{f(m)}$-agents in a stable configuration.

\begin{lemma}
\label{lem:impsink}
When a sink state $m$ exists in $Q_p$, for any reduced weakly-fair execution $E = C_0$, $C_1$, $C_2$, $\ldots$, $C_t$, $\ldots$ of $Alg_{sym}$ with $P-1$ agents and an initialized BS, the following is satisfied in a stable configuration $C_t$ of $E$.
\begin{itemize}
\item The number of $f(m)$-agents is $P/2-1$ and the number of $\overline{f(m)}$-agents is $P/2$.
\end{itemize}

\end{lemma}

\begin{proof}
The idea of the proof is as follows.
For contradiction, we consider a reduced execution $E$ with $P-1$ agents such that the number of $f(m)$-agents is $P/2$ and the number of $\overline{f(m)}$-agents is $P/2-1$ in a stable configuration of $E$.
Note that, by Lemma \ref{lem:P-24}, the number of $f(m)$-states is $P/2$. Thus, since the stable configuration is reduced, some agent has $m$ in the stable configuration.
Next, consider an execution with $P$ agents such that one additional agent has $m$ as an initial state and other agents behave similarly to $E$.
In the execution, the additional agent does not join the interactions until $P-1$ agents converge to a stable configuration in $E$. 
At that time, one of the $P-1$ agents has state $m$ and the additional agent also has state $m$. 
We can prove that, from this configuration, $P-1$ agents cannot recognize the additional agent and hence they make the same behavior as in $E$.
In addition, the additional agent can keep state $m$.
Since the additional agent has $f(m)$-state, the number of $f(m)$-agents is $P/2+1$ and the number of $\overline{f(m)}$-agents is $P/2-1$. 
This implies that the uniform bipartition problem cannot be solved.

From now, we show the details of proof.
For contradiction, there exists a reduced weakly-fair execution $E$ such that the number of $f(m)$-agents is $P/2$ and the number of $\overline{f(m)}$-agents is $P/2-1$ in a stable configuration of $E$.

First, consider a population $A = \{ a_0, a_1, a_2, \ldots, a_{P-1} \}$ of $P-1$ agents and an initialized BS, where $a_0$ is the BS.
We define a reduced weakly-fair execution $E = C_0$, $C_1$, $C_2$, $\ldots$, $C_t$, $\ldots$, $C_{t'_0}, e_{1}$, $C_{t'_1}, e_{2}$, $C_{t'_2}, e_{3}$, $\ldots$ of $Alg_{sym}$ with $A$ as follows.
\begin{itemize}
\item $C_t$ is a stable configuration.
\item For any $u \ge 0$, $C_{t'_u}$ is a particular stable reduced configuration such that $C_{t'_0} = C_{t'_1} = C_{t'_2} = \cdots$ holds. Note that such a configuration (i.e., a stable reduced configuration that appears infinite number of times) exists because the number of possible configurations is finite.
\item For $j > 0$, $e_{j}$ is an execution segment such that, during execution segment $C_{t'_{j-1}},e_{j},C_{t'_j}$, any pair of agents in $A$ interact at least once. This is possible because $E$ is weakly-fair.
\end{itemize}
By Lemma \ref{lem:P-24}, the number of $f(m)$-states is $P/2$.
In addition, by the assumption, the number of $f(m)$-agents is also $P/2$ in $C_{t'_u}$ for any $u$.
From these facts, for any $u$, since $C_{t'_u}$ is a reduced configuration (i.e., there exist no homonyms except $m$), there exists a particular agent $a_q$ that has state $m$ in $C_{t'_u}$ for any $u$.

Next, consider a population $A' = \{a'_0, a'_1, \ldots$, $a'_P\}$ of $P$ agents and an initialized BS, where $a'_0$ is the BS.
We define $E' = C'_0$, $C'_1$, $C'_2$, $\ldots$, $C'_t$, $\ldots$, $C'_{t'_0}, e'_{1}$, $C'_{t'_1}, e'_{2}$, $C'_{t'_2}, e'_{3}$, $\ldots$ by using $E$.
First, we define the first part of $E'$, that is, $C'_0$, $C'_1$, $C'_2$, $\ldots$, $C'_t$, $\ldots$, $C'_{t'_0}$ as follows:
\begin{itemize}
\item In initial configuration $C'_0$, $a'_0,\ldots,a'_{P-1}$ have the same states as $a_0,\ldots,a_{P-1}$ in $C_0$ and $a'_P$ has state $m$. Formally, $s(a'_i, C'_0) = s(a_i, C_0)$ holds for $i$ ($P-1\ge i\ge 0$), and $s(a'_P, C'_0) = m$ holds.
\item From $C'_0$ to $C'_{t'_0}$, $a'_0,\ldots,a'_{P-1}$ interact similarly to $a_0,\ldots,a_{P-1}$ in $E$. Formally, for any $u$($t'_0>u\ge 0$), when $a_g$ and $a_h$ interact at $C_u \rightarrow C_{u+1}$, $a'_g$ and $a'_h$ interact at $C'_u \rightarrow C'_{u+1}$.
\end{itemize}

Clearly, $s(a'_i, C'_{t'_0}) = s(a_i, C_{t'_0})$ holds for $i$ ($P-1 \ge i \ge 0$) and $s(a'_P, C'_{t'_0}) = m$ hold.
This implies that the number of $f(m)$-agents is $P/2+1$ and the number of $\overline{f(m)}$-agents is $P/2-1$ in $C'_{t'_0}$.

We define the remaining part of $E'$, that is, $C'_{t'_0}, e'_{1}$, $C'_{t'_1}, e'_{2}$, $C'_{t'_2}, e'_{3}$, $\ldots$ as follows:
\begin{itemize}
\item For $j > 0$, we construct an execution segment $e'_j = \Ch{j}{1}$, $\Ch{j}{2}$, $\Ch{j}{3}$, $\ldots$, $\Ch{j}{z}$, $\Ch{j}{z+1}$ by using $e_{j} = \C{j}{1}$, $\C{j}{2}$, $\C{j}{3}$, $\ldots$, $\C{j}{z}$ , where $z=|e_{j}|$ holds.
Concretely, we construct $C'_{t'_{j-1}}$, $e'_{j}$, $C'_{t'_j}$ as follows:
\begin{itemize}
\item Case that $j$ is even. In this case, agents $a'_0,\ldots,a'_{P-1}$ interact in execution segment $C'_{t'_{j-1}},e'_j$ similarly to $a_0,\ldots,a_{P-1}$ in execution segment $C_{t'_{j-1}},e_j,C_{t'_j}$. Formally, when $a_g$ and $a_h$ interact at $\C{j}{f} \rightarrow \C{j}{f+1}$ for $z>f>0$ (resp., $C_{t'_{j-1}} \rightarrow \C{j}{1}$ and $\C{j}{z} \rightarrow C_{t'_j}$), $a'_g$ and $a'_h$ interact at $\Ch{j}{f} \rightarrow \Ch{j}{f+1}$ (resp., $C'_{t'_{j-1}} \rightarrow \Ch{j}{1}$ and $\Ch{j}{z} \rightarrow \Ch{j}{z+1}$).
\item Case that $j$ is odd. In this case, $a'_P$ joins interactions instead of $a'_q$. Note that, in $C'_{t'_{j-1}}$, both $a'_P$ and $a'_q$ have state $m$. Formally we construct $e'_j$ as follows:
(1) when $a_g(g\neq q)$ and $a_h(h\neq q)$ interact at $\C{j}{f} \rightarrow \C{j}{f+1}$ for $z>f>0$ (resp., $C_{t'_{j-1}} \rightarrow \C{j}{1}$ and $\C{j}{z} \rightarrow C_{t'_j}$), $a'_g$ and $a'_h$ interact at $\Ch{j}{f} \rightarrow \Ch{j}{f+1}$ (resp., $C'_{t'_{j-1}} \rightarrow \Ch{j}{1}$ and $\Ch{j}{z} \rightarrow \Ch{j}{z+1}$), 
(2) when $a_q$ interacts with an agent $a_i(i\neq q)$ at $\C{j}{f} \rightarrow \C{j}{f+1}$ (resp., $C_{t'_{j-1}} \rightarrow \C{j}{1}$ and $\C{j}{z} \rightarrow C_{t'_j}$), $a'_{P}$ interacts with $a'_i$ at $\Ch{j}{f} \rightarrow \Ch{j}{f+1}$ (resp., $C'_{t'_{j-1}} \rightarrow \Ch{j}{1}$ and $\Ch{j}{z} \rightarrow \Ch{j}{z+1}$).
\item $a'_P$ and $a'_q$ interact at $\Ch{j}{z+1} \rightarrow C'_{t'_{j}}$.
\end{itemize}
\end{itemize}

We can inductively show that, for any $x\ge 0$, $s(a'_i, C'_{t'_x}) = s(a_i, C_{t'_x})$ holds for any $i$ ($P-1 \ge i \ge 0$) and $s(a'_q, C'_{t'_x}) = s(a'_P, C'_{t'_x}) = s(a_q, C_{t'_x}) = m $ holds. 
Clearly this holds for $x=0$. Assume that this holds for $x=j$, and consider the case of $x=j+1$. 
When $j+1$ is even, during execution segment $C'_{t'_{j}}$, $e'_{j+1}$, agents in $A'-\{a'_P \}$ interact similarly to $C_{t'_{j}}$, $e_{j+1}$, $C_{t'_{j+1}}$.
Hence, for any $j > 0$, $s(a'_i, \Ch{j+1}{z+1}) = s(a_i, C_{t'_{j+1}})$ holds for any $i$ ($P-1\ge i\ge 0$) and $s(a'_q, \Ch{j+1}{z+1}) = s(a'_P, \Ch{j+1}{z+1}) = s(a_q, C_{t'_{j+1}}) = m $ holds.
At $\Ch{j+1}{z+1} \rightarrow C'_{t'_{j+1}}$, $a'_P$ and $a'_q$ interact and hence, if they have state $m$, they do not change their states.
Thus, $s(a'_i, C'_{t'_{j+1}}) = s(a_i, C_{t'_{j+1}})$ holds for any $i$ ($P-1 \ge i \ge 0$) and $s(a'_q, C'_{t'_{j+1}}) = s(a'_P, C'_{t'_{j+1}}) = s(a_q, C_{t'_{j+1}}) = m $ holds.
When $j+1$ is odd, during execution segment $C'_{t'_{j}}$, $e'_{j+1}$, $a'_P$ joins interactions instead of $a'_q$ and agents in $A'-\{a'_q \}$ behave similarly to $C_{t'_{j-1}}$, $e_j$, $C_{t'_j}$.
Hence, for any $j > 0$, $s(a'_i, \Ch{j+1}{z+1}) = s(a_i, C_{t'_{j+1}})$ holds for any $i$ ($P-1\ge i\ge 0$) and $s(a'_q, \Ch{j+1}{z+1}) = s(a'_P, \Ch{j+1}{z+1}) = s(a_q, C_{t'_{j+1}}) = m $ holds.
At $\Ch{j+1}{z+1} \rightarrow C'_{t'_{j+1}}$, $a'_P$ and $a'_q$ interact and hence, if they have state $m$, they do not change their states.
Thus, $s(a'_i, C'_{t'_{j+1}}) = s(a_i, C_{t'_{j+1}})$ holds for any $i$ ($P-1 \ge i \ge 0$) and $s(a'_q, C'_{t'_{j+1}}) = s(a'_P, C'_{t'_{j+1}}) = s(a_q, C_{t'_{j+1}}) = m $ holds.

For any $j$, the number of $f(m)$-agents is $P/2$ and the number of $\overline{f(m)}$-agents is $P/2-1$ in $C_{t'_j}$, and thus, the number of $f(m)$-agents is $P/2+1$ and the number of $\overline{f(m)}$-agents is $P/2-1$ in $C'_{t'_j}$.
Therefore, $E'$ cannot solve the uniform bipartition.
During $C'_{t'_{j-1}}$, $e'_j$, when $j$ is even, agents in $A'-\{a'_{P}\}$ interact each other.
Similarly,  when $j$ is odd, agents in $A'-\{a'_q\}$ interact each other.
Moreover, for $j > 0$, at $\Ch{j}{z+1} \rightarrow C'_{t'_{j}}$, $a'_q$ and $a'_P$ interact.
Thus, $E'$ is weakly-fair.
Since $E'$ cannot solve the uniform bipartition, this is a contradiction.
\end{proof}


By using these lemmas, we show that, when a sink state exists in $Q_p$, $Alg_{sym}$ does not work.
We prove this in a similar way to the case of naming protocols in \cite{burman2018brief}, but we need a non-trivial modifications to apply the proof to uniform bipartition protocols.

\begin{definition}
Assume that a sink state $m$ exists in $Q_p$.
Consider configurations $C_0$ and $C_1$ for a population $A$.
We say that $C_0$ is far away from $C_1$ by a non-sink state $s \neq m$ if there exists an agent $a_x$ such that $s(a_y, C_0) = s(a_y, C_1)$ for any $a_y\in A-\{a_x\}$, $s(a_x, C_0) = m$, and $s(a_x, C_1) = s \neq m$ hold.
Then, we denote $C_0$ as $C^{-s}_1$ and denote $C_1$ as $C^{+s}_0$.
\end{definition}

Here we introduce the notion of equivalent configurations. We say that configurations $C$ and $C'$ are equivalent if a multi-set of states in $C$ is identical to that in $C'$.

\begin{lemma}
\label{lem:name8}
Assume that a sink state $m$ exists in $Q_p$. Consider an execution segment $e = C_0, C_1, C_2, \ldots, C_k$ of $Alg_{sym}$ with $P$ agents and an initialized BS, that satisfies the following conditions:
\begin{itemize}
\item $e$ is a reduced execution segment.
\item $C_0$ is reduced.
\item There exists a non-sink state $s$ such that, in any reduced configuration of $e$ except the last configuration $C_k$, no agent has state $s$.
\end{itemize}
Then, there exists the execution segment $e'= C'_0, C'_1, C'_2, \ldots, C'_k$ of $Alg_{sym}$ with $P$ agents and an initialized BS, that satisfies the following conditions:
\begin{itemize}
\item A particular agent $a_x$ with $m$ does not join interactions.
\item $C_0=C'_0$ holds.
\item For any $i$ ($0<i\le k$), $C'_i$ and $C_i$ are equivalent.
\end{itemize}
\end{lemma}
\begin{proof}
In a reduced configuration with $P$ agents, if there exists a non-sink state that is held by no agent, there are at least two agents with a sink state.
Hence, in any reduced configuration of $e$ except the last configuration $C_k$, there exist at least two agents with a sink state.
By using this property, we construct $e'$ by induction on the index of configuration.
First, since we can set the initial configuration $C'_0$ such that $C'_0 = C_0$ holds, the base case holds.

For the induction step, assume that there exists a configuration $C'_l (l \ge 0)$ such that $C'_u$ and $C_u$ are equivalent for any $u\le l$ and $a_x$ does not join interactions until $C'_l$ (i.e., $a_x$ has a sink state in $C'_l$). 
We consider two cases of interaction at $C_l \rightarrow C_{l+1}$.

First, we consider the case that an agent with a sink state does not join an interaction at $C_l \rightarrow C_{l+1}$.
In this case, since $C_l$ and $C'_l$ are equivalent and $a_x$ has a sink state in $C'_l$, a state transition that happens at $C_l \rightarrow C_{l+1}$ can happen at $C'_l \rightarrow C'_{l+1}$. Thus, the lemma holds at $C'_{l+1}$.

Next, we consider the case that an agent with a sink state interacts at $C_l \rightarrow C_{l+1}$.
In the case, $C_l$ and $C'_l$ are reduced.
By the assumption, in $C_l$ and $C'_l$, there are at least two agents with a sink state.
Let $a_y \neq a_x$ be an agent that has a sink state in $C'_l$.
Then, when agents $a_i$ and $a_j$ interact at $C_l \rightarrow C_{l+1}$, we consider the following two cases.
\begin{itemize}
\item Case that $a_i$ and $a_j$ have a sink state:
In this case, $C_l = C_{l+1}$ holds.
Hence, we skip this interaction and regard $C'_l$ as $C'_{l+1}$\footnote{Strictly speaking, this violates the definition of an execution because no interaction happens at $C'_l\rightarrow C'_{l+1}$. However, by removing $C'_{l+1}$ from $e'$, we can construct execution $e'$ that satisfies the definition of an execution. This modification does not effect the following proofs.}. Clearly, $C'_{l+1}$ and $C_{l+1}$ are equivalent.

\item Case that either $a_i$ or $a_j$ has a sink state:
Without loss of generality, we assume that $a_i$ has a sink state.
In this case, by making interaction between $a_y$ and $a_j$ at $C'_l \rightarrow C'_{l+1}$, we can obtain $C'_{l+1}$ such that $C'_{l+1}$ and $C_{l+1}$ are equivalent.
\end{itemize}

Then, we can obtain $C'_{l+1}$ without making $a_x$ join an interaction.
Thus, the lemma holds.
\end{proof}

The following lemma is identical to the lemma in \cite{burman2018brief}. Although the lemma is proved for naming protocols in \cite{burman2018brief}, we can use the lemma because the proof does not use the property of naming protocols. For completeness, we also give the proof.

\begin{lemma}
\label{lem:name9}
Assume that a sink state $m$ exists in $Q_p$. Consider two reduced configurations $C_0$ and $C^{-s}_0$ with $P$ agents.
The difference between $C_0$ and $C^{-s}_0$ is only whether an agent $a_x$ has a non-sink state $s$ or a sink state $m$.
Consider a reduced execution segment $C^{-s}_0, e^{-s}_0, C_1$ of $Alg_{sym}$ with $P$ agents and an initialized BS, that satisfies the following conditions:
\begin{itemize}
\item During $C^{-s}_0, e^{-s}_0, C_1$, $a_x$ does not join interactions.
\item In any reduced configuration during $C^{-s}_0, e^{-s}_0$, there exists no agent with $s$.
\item $C_1$ is a reduced configuration such that there exists an agent with $s$.
\end{itemize}
If there exists such an execution segment, there also exists the reduced execution segment $C_0, e_0, C^{-s}_1$ of $Alg_{sym}$ with $P$ agents and an initialized BS, that satisfies the following conditions:
\begin{itemize}
\item During $C_0, e_0, C^{-s}_1$, $a_x$ does not join interactions except last reducing.
\item In any reduced configuration during $C_0, e_0$, there exists an agent with $s$.
\item $C^{-s}_1$ is a reduced configuration such that there does not exist an agent with $s$.
\end{itemize}
\end{lemma}
\begin{proof}
By making interaction similar to $C^{-s}_0, e^{-s}_0$, $C_1$, we can construct a reduced execution segment $C_0, e'_0$, $C^{+s}_1$ starting from $C_0$.
Since exactly two agents have state $s$ in $C^{+s}_1$, $C^{+s}_1$ can be reduce to $C^{-s}_1$. 
We denote this reducing procedure as $C^{+s}_1$, $e^r$, $C^{-s}_1$.
Then, we can obtain a execution segment $C_0$, $e'_0$, $C^{+s}_1$, $e^r$, $C^{-s}_1$.
Note that, during $C_0, e$, $C^{+s}_1$, agent $a_x$ does not join interactions. This implies that, during $C_0$, $e'_0$, $C^{+s}_1$, $e^r$, $C^{-s}_1$, agent $a_x$ does not join interactions except last reducing.
Hence, we can obtain the required execution segment $C_0, e_0, C^{-s}_1$ such that $e_0 = e'_0$, $C^{+s}_1$, $e^r$ holds.
\end{proof}

In the next lemma, we show that a sink state does not exist.

\begin{lemma}
\label{lem:impsink1}
A sink state does not exist in $Q_p$.
\end{lemma}
\begin{proof}
We use the idea of the impossibility proof for the naming protocol \cite{burman2018brief}.
The idea of the proof is as follows. 
For contradiction, assume that there exists a sink state $m \in Q_p$.
First, consider an execution segment $e$ with $P$ agents such that, 1) a particular agent $a_x$ does not join interactions and other $P-1$ agents interact until convergence, 2) $a_x$ has $m$ as an initial state, and 3) its final configuration $C_h$ is a reduced configuration.
Let $s$ be a $\overline{f(m)}$-state. When $a_x$ has $s$ as an initial state, by making other agents interact similarly to $e$ we can obtain $C^{+s}_h$. 
Moreover, by Lemma \ref{lem:impsink}, since the number of $\overline{f(m)}$-agents except for $a_x$ is $P/2$ in $C^{+s}_h$ and $C_h$ is reduced, there exists an agent with $s$ except for $a_x$ in $C^{+s}_h$.
Thus, $C^{-s}_h$ can be obtained by reducing $C^{+s}_h$.
Since every $\overline{f(m)}$-state must be held by one agent in a stable reduced configuration, $C^{-s}_h$ is not stable.
Thus, by Lemma \ref{lem:name8}, we can construct $C^{-s}_h$, $e^{-s}_h$, $C_{h+1}$ such that an agent $a_y \neq a_x$ with $m$ does not join interactions and there exists an agent with $s$ in $C_{h+1}$.
After that, by repeating the application of Lemma \ref{lem:name9}, we can construct a weakly-fair execution of $Alg_{sym}$ so that an agent with $s$ disappears infinitely often.
Since every $\overline{f(m)}$-state must be held by one agent in a stable reduced configuration, such execution cannot solve the uniform bipartition and thus the lemma holds.

From now on, we show the details of the proof.
For contradiction, there exists a sink state $m \in Q_p$.

Consider a population $A = \{ a_0, a_1, a_2, \ldots, a_{P} \}$ of $P$ agents and an initialized BS, where $a_0$ is the BS.
First, we consider an execution segment $e = C_0$, $C_1$, $C_2$, $\ldots$, $C_h$ such that, 
\begin{itemize}
\item a particular agent $a_x$ has state $m$ in $C_0$ and does not interact during $e$. 
\item other $P-1$ agents (and the BS) interact until convergence, which implies that, by Lemma \ref{lem:impsink}, the number of $f(m)$-agents in $A-\{a_x\}$ is $P/2-1$ and the number of $\overline{f(m)}$-agents in $A-\{a_x\}$ is $P/2$ after some configuration of $e$, and, 
\item $C_h$ is a reduced configuration.
\end{itemize}
In $C_h$, by Lemma \ref{lem:impsink}, the number of $f(m)$-agents is $P/2$ and the number of $\overline{f(m)}$-agents is $P/2$ because additional agent $a_x$ has state $m$.
Note that all $\overline{f(m)}$-states are non-sink states and, by Lemma \ref{lem:P-24}, the number of $\overline{f(m)}$-states is $P/2$.
Since $C_h$ is reduced, no two agents have the same non-sink state and thus every $\overline{f(m)}$-state is held by exactly one agent in $C_h$.

Let $s$ be a $\overline{f(m)}$-state. We consider three configurations $\Cu{0} = C_h$, $\Cuu{+s}{0}$, and $\Cuu{-s}{0}$.
Since $a_x$ does not interact in $e$, $\Cuu{+s}{0}$ can be obtained by the same interactions in $e$ if $a_x$ has state $s$ in the initial configuration (Note that this execution may not be a reduced execution).
In $\Cuu{+s}{0}$, since every $\overline{f(m)}$-state is held by exactly one agent in $A - \{a_x\}$, there exists exactly one agent $a_s \neq a_x$ with $s$.
Hence, we can obtain a reduced configuration $\Cuu{-s}{0}$ by reducing $\Cuu{+s}{0}$.
Note that, since the number of $\overline{f(m)}$-states is $P/2$, every $\overline{f(m)}$-state must be held by one agent in any stable reduced configuration with $P$ agents.
Hence, since $\Cuu{-s}{0}$ is reduced and no agent has state $s$ in $\Cuu{-s}{0}$, $\Cuu{-s}{0}$ is not stable.
Hence, there exists an execution segment from $C^{-s}_{u_0}$ that leads to a stable configuration where some agent has state $s$. This implies that, we can construct a reduced execution segment $\epsilon_1 = \Cuu{-s}{0}$, $e^{-s}_0$, $\Cu{1}$ of $Alg_{sym}$ starting from $\Cuu{-s}{0}$ as follows:
\begin{itemize}
\item $\Cu{1}$ is reduced and exactly one agent $a_y$ has state $s$ in $\Cu{1}$.
\item For any reduced configuration in $e^{-s}_0$, no agent has state $s$.
\end{itemize}
Moreover, by Lemma \ref{lem:name8}, since $a_x$ has $m$ in $\Cuu{-s}{0}$, we can construct $\epsilon_1$ such that $a_x$ does not join interactions.
Hence, by Lemma \ref{lem:name9}, we can construct a reduced execution segment $\epsilon'_1 = \Cu{0}$, $e_0$, $\Cuu{-s}{1}$ of $Alg_{sym}$. Note that, in $\Cuu{-s}{1}$, $a_y$ has state $m$.

Similarly, we can construct a reduced execution segment $\epsilon_2 = \Cuu{-s}{1}$, $e^{-s}_1$, $\Cu{2}$ such that, 
\begin{itemize}
\item $\Cu{2}$ is reduced and exactly one agent $a_z$ has state $s$ in $\Cu{2}$, 
\item for any reduced configuration in $e^{-s}_1$, no agent has state $s$.
\end{itemize}
By Lemma \ref{lem:name8}, we can construct $\epsilon_2$ such that $a_y$ does not join interactions. 
Hence, by Lemma \ref{lem:name9}, we can construct a reduced execution segment $\epsilon'_2 = \Cu{1}$, $e_1$, $\Cuu{-s}{2}$ of $Alg_{sym}$.

By repeating similar arguments, we can construct an infinite execution segment $\epsilon^* = \Cuu{-s}{0}$, $e^{-s}_0$, $\Cu{1}$, $e_1$, $\Cuu{-s}{2}$, $e^{-s}_2$, $\Cu{3}$, $\ldots$.
Recall that, for $i \ge 0$, $\Cuu{-s}{i}$ is not stable because every $\overline{f(m)}$-state must be held by one agent in a stable reduced configuration.
Thus, $\epsilon^*$ cannot reach a stable configuration.
As described above, there exists an execution segment $e_{ini}$ that reaches $C^{-s}_{u_0}$ from some initial configuration. Hence, we can construct an execution $E = e^{ini}$, $\Cuu{-s}{0}$, $e^{-s}_0$, $\Cu{1}$, $e_1$, $\Cuu{-s}{2}$, $e^{-s}_2$, $\Cu{3}$, $\ldots$ that does not reach a stable configuration. The remaining thing is to show that we can construct $E$ so that $E$ is weakly-fair. 

Recall the way to construct an execution segment $\epsilon_i=\Cuu{-s}{i-1}$, $e^{-s}_{i-1}$, $\Cu{i}$. Since $C^{-s}_{u_i}$ is reduced, any pair of agents can interact at the first interaction of $\epsilon_i$. Consequently, we can construct $\epsilon_1,\epsilon_2,\ldots$ so that every pair of agents interact infinite number of times in the first interactions of $\epsilon_1,\epsilon_2,\ldots$. In addition, a pair of agents that interact in $\epsilon_i$ also interact in $\epsilon'_i$. Hence, we can construct $E$ so that every pair of agents interact infinite number of times. This implies that $E$ is weakly-fair, but $E$ cannot solve the problem. This is a contradiction.
\end{proof}

Finally, we show that, even if a sink pair exists in $Q_p$, $Alg_{sym}$ does not work.

\begin{lemma}
\label{lem:impsink2}
A sink pair does not exist in $Q_p$.
\end{lemma}
\begin{proof}
For contradiction, assume that there exists a sink pair in $Q_p$.

Let $m_1$ and $m_2$ be a sink pair. 
We consider two cases: 1) $f(m_1)=f(m_2)$ holds, or 2) $f(m_1)\neq f(m_2)$ holds.

First, we consider the case that $f(m_1) = f(m_2)$ holds.
By Corollary \ref{cor:P-23}, the number of $red$ (resp., $blue$) states not in the sink pair is $P/2-1$ (resp., $P/2-1$). Since we assume $|Q_{blue}|\le|Q_{red}|$, $f(m_1)=f(m_2)=red$ holds.
Hence, $|Q_{blue}|=P/2-1$ and $|Q_{red}|=P/2+1$ hold.
Consider a reduced weakly-fair execution $E^*$ of $Alg_{sym}$ with $P$ agents and an initialized BS.
Since $E^*$ is a reduced execution, a stable reduced configuration occurs infinitely often in $E^*$.
Every state not in the sink pair is held by at most one agent in any reduced configuration, and thus, at most $P/2-1$ $blue$ agents exist in any reduced configuration in $E^*$.
Hence, any reduced configuration in $E^*$ is not stable.
This is a contradiction.
Therefore, $f(m_1) = f(m_2)$ does not hold (i.e., $f(m_1) \neq f(m_2)$ holds).

Next, we consider the case that $f(m_1) \neq f(m_2)$ holds.
The idea of the proof is as follows.
We consider an execution $E$ with $P-1$ agents.
Let $m^*$ be a state in the sink pair such that the number of $f(m^*)$-agents is $P/2$ in a stable reduced configuration of $E$. 
This implies that there exists some agent with $m^*$ in the configuration.
Next, consider an execution with $P$ agents such that one additional agent has $m^*$ as an initial state and other agents behave similarly to $E$.
In the execution, the additional agent does not join the interactions until $P-1$ agents converge to a stable reduced configuration in $E$. 
At that time, one of the $P-1$ agents has state $m^*$ and the additional agent also has state $m^*$. 
We can prove that, from this configuration, $P-1$ agents cannot recognize the additional agent and hence they make the same behavior as in $E$.
In addition, the additional agent can keep state $m^*$.
Since the additional agent has $f(m^*)$-state, the number of $f(m^*)$-agents is $P/2+1$ and the number of $\overline{f(m^*)}$-agents is $P/2-1$. 
This implies that the uniform bipartition problem cannot be solved.

From now on, we show the details of the proof.
Without loss of generality, we assume that $f(m_1)=red$ and $f(m_2) = blue$ hold.
Consider a population $A = \{ a_0, a_1, a_2, \ldots, a_{P-1} \}$ of $P-1$ agents and an initialized BS, where $a_0$ is the BS.
We define a reduced weakly-fair execution $E = C_0$, $C_1$, $C_2$, $\ldots$, $C_t$, $\ldots$, $C_{t'_0}, e_{1}$, $C_{t'_1}, e_{2}$, $C_{t'_2}, e_{3}$, $\ldots$ of $Alg_{sym}$ with $A$ as follows.
\begin{itemize}
\item $C_t$ is a stable configuration.
\item For any $u \ge 0$, $C_{t'_u}$ is a particular stable reduced configuration such that $C_{t'_0} = C_{t'_1} = C_{t'_2} = \cdots$ holds. 
Note that such a configuration (i.e., a stable reduced configuration that appears infinite number of times) exists because the number of possible configurations is finite.
\item For $j > 0$, $e_{j}$ is an execution segment such that, during execution segment $C_{t'_{j-1}},e_{j},C_{t'_j}$, any pair of agents in $A$ interact at least once. This is possible because $E$ is weakly-fair.
\end{itemize}
Let $m^*$ be a state in the sink pair such that the number of $f(m^*)$-agents is $P/2$ in $C_{t'_0}$. Note that $m^*$ is uniquely decided because the number of agents is $P-1$. Since every state not in the sink pair is held by at most one agent in a reduced configuration, there exists an agent $a_q$ that has state $m^*$ in $C_{t'_0}$.

Subsequently, consider a population $A' = \{a'_0$, $a'_1$, $a'_2, \ldots$, $a'_{P} \}$ of $P$ agents and an initialized BS, where $a'_0$ is the BS.
We define a reduced weakly-fair execution $E' = C'_0$, $C'_1$, $C'_2$, $\ldots$, $C'_t$, $\ldots$ $C'_{t'_0}$, $e'_{1}$, $C'_{t'_1}$, $e'_{2}$, $C'_{t'_2}$, $e'_{3}$, $\ldots$ of $Alg_{sym}$ with $A'$.
First, we define the first part of $E'$, that is, $C'_0$, $C'_1$, $C'_2$, $\ldots$, $C'_t$, $\ldots$ $C'_{t'_0}$ as follows.
\begin{itemize}
\item In initial configuration $C'_0$, $a'_0,\ldots,a'_{P-1}$ have the same states as $a_0,\ldots,a_{P-1}$ in $C_0$ and $a'_P$ has state $m^*$. Formally, $s(a'_i, C'_0) = s(a_i, C_0)$ holds for $i$ ($P-1\ge i\ge 0$), and $s(a'_P, C'_0) = m^*$ holds.
\item From $C'_0$ to $C'_{t'_0}$, $a'_0,\ldots,a'_{P-1}$ interact similarly to $a_0,\ldots,a_{P-1}$ in $E$. Formally, for any $u$($t'_0>u\ge 0$), when $a_g$ and $a_h$ interact at $C_u \rightarrow C_{u+1}$, $a'_g$ and $a'_h$ interact at $C'_u \rightarrow C'_{u+1}$.
\end{itemize}
Clearly, $s(a'_i, C'_{t'_0}) = s(a_i, C_{t'_0})$ holds for $i$ ($P-1 \ge i \ge 0$) and $s(a'_P, C'_{t'_0}) = m^*$ hold.
This implies that the number of $f(m^*)$-agents is $P/2+1$ and the number of $\overline{f(m^*)}$-agents is $P/2-1$ in $C'_{t'_0}$.

Then, we construct the remaining part of $E'$, that is, $C'_{t'_0}$, $e'_{1}$, $C'_{t'_1}$, $e'_{2}$, $C'_{t'_2}$, $e'_{3}$, $\ldots$ as follows:
\begin{itemize}
\item For $j > 0$, we construct an execution segment $e'_{j} = \Ch{j}{1}$, $\Ch{j}{2}$, $\Ch{j}{3}$, $\ldots$, $\Ch{j}{z}$, $\Ch{j}{z+1}$, $\Ch{j}{z+2}$ by using $e_{j} = \C{j}{1}$, $\C{j}{2}$, $\C{j}{3}$, $\ldots$, $\C{j}{z}$, where $z=|e_{j}|$ holds.
Concretely, we construct $C'_{t'_{j-1}}$, $e'_{j}$, $C'_{t'_j}$ as follows:
\begin{itemize}
\item Case that $j$ is even. In this case, agents $a'_0,\ldots,a'_{P-1}$ interact in execution segment $C'_{t'_{j-1}},\Ch{j}{1}$, $\Ch{j}{2}$, $\ldots$, $\Ch{j}{z+1}$ similarly to $a_0,\ldots,a_{P-1}$ in execution segment $C_{t'_{j-1}},e_j,C_{t'_j}$. Formally, when $a_g$ and $a_h$ interact at $\C{j}{f} \rightarrow \C{j}{f+1}$ for $z>f>0$ (resp., $C_{t'_{j-1}} \rightarrow \C{j}{1}$ and $\C{j}{z} \rightarrow C_{t'_j}$), $a'_g$ and $a'_h$ interact at $\Ch{j}{f} \rightarrow \Ch{j}{f+1}$ (resp., $C'_{t'_{j-1}} \rightarrow \Ch{j}{1}$ and $\Ch{j}{z} \rightarrow \Ch{j}{z+1}$).
\item Case that $j$ is odd. In this case, $a'_P$ joins interactions instead of $a'_q$. Note that, in $C'_{t'_{j-1}}$, both $a'_P$ and $a'_q$ have state $m^*$. Formally we construct $e'_j$ as follows:
(1) when $a_g(g\neq q)$ and $a_h(h\neq q)$ interact at $\C{j}{f} \rightarrow \C{j}{f+1}$ for $z>f>0$ (resp., $C_{t'_{j-1}} \rightarrow \C{j}{1}$ and $\C{j}{z} \rightarrow C_{t'_j}$), $a'_g$ and $a'_h$ interact at $\Ch{j}{f} \rightarrow \Ch{j}{f+1}$ (resp., $C'_{t'_{j-1}} \rightarrow \Ch{j}{1}$ and $\Ch{j}{z} \rightarrow \Ch{j}{z+1}$), 
(2) when $a_q$ interacts with an agent $a_i(i\neq q)$ at $\C{j}{f} \rightarrow \C{j}{f+1}$ (resp., $C_{t'_{j-1}} \rightarrow \C{j}{1}$ and $\C{j}{z} \rightarrow C_{t'_j}$), $a'_{P}$ interacts with $a'_i$ at $\Ch{j}{f} \rightarrow \Ch{j}{f+1}$ (resp., $C'_{t'_{j-1}} \rightarrow \Ch{j}{1}$ and $\Ch{j}{z} \rightarrow \Ch{j}{z+1}$).
\item $a'_P$ and $a'_q$ interact at $\Ch{j}{z+1} \rightarrow \Ch{j}{z+2}$. Also, $a'_P$ and $a'_q$ interact at $\Ch{j}{z+2} \rightarrow C'_{t'_{j}}$.
\end{itemize}
\end{itemize}

We can inductively show that, for any $x\ge 0$, $s(a'_i, C'_{t'_x}) = s(a_i, C_{t'_x})$ holds for any $i$ ($P-1 \ge i \ge 0$) and $s(a'_q, C'_{t'_x}) = s(a'_P, C'_{t'_x}) = s(a_q, C_{t'_x}) = m^* $ holds. 
Clearly this holds for $x=0$. Assume that this holds for $x=j$, and consider the case of $x=j+1$. 
When $j+1$ is even, during execution segment $C'_{t'_{j}},\Ch{j+1}{1}$, $\Ch{j+1}{2}$, $\ldots$, $\Ch{j+1}{z+1}$, agents in $A'-\{a'_P \}$ interact similarly to $C_{t'_{j}}$, $e_{j+1}$, $C_{t'_{j+1}}$.
Hence, $s(a'_i, \Ch{j+1}{z+1}) = s(a_i, C_{t'_{j+1}})$ holds for any $i$ ($P-1\ge i\ge 0$) and $s(a'_q, \Ch{j+1}{z+1}) = s(a'_P, \Ch{j+1}{z+1}) = s(a_q, C_{t'_{j+1}}) = m^* $ holds.
$a'_P$ and $a'_q$ interact at $\Ch{j+1}{z+1} \rightarrow \Ch{j+1}{z+2}$ and $\Ch{j+1}{z+2} \rightarrow C'_{t'_{j+1}}$.
By the assumption of $m^*$, if two agents with $m^*$ interact twice, they keep their state $m^*$. 
Since $a'_P$ and $a'_q$ have state $m^*$ in $\Ch{j+1}{z+1}$, they also have state $m^*$ in $C'_{t'_{j+1}}$, and thus, $\Ch{j+1}{z+1}$ is equal to $C'_{t'_{j+1}}$.
Thus, $s(a'_i, C'_{t'_{j+1}}) = s(a_i, C_{t'_{j+1}})$ holds for any $i$ ($P-1 \ge i \ge 0$) and $s(a'_q, C'_{t'_{j+1}}) = s(a'_P, C'_{t'_{j+1}}) = s(a_q, C_{t'_{j+1}}) = m^* $ holds.
When $j+1$ is odd, during execution segment $C'_{t'_{j}},\Ch{j+1}{1}$, $\Ch{j+1}{2}$, $\ldots$, $\Ch{j+1}{z+1}$, $a'_P$ interacts instead of $a'_q$ and agents in $A'-\{a'_q \}$ behave similarly to $C_{t'_{j-1}}$, $e_j$, $C_{t'_j}$.
Hence, $s(a'_i, \Ch{j+1}{z+1}) = s(a_i, C_{t'_{j+1}})$ holds for any $i$ ($P-1\ge i\ge 0$) and $s(a'_q, \Ch{j+1}{z+1}) = s(a'_P, \Ch{j+1}{z+1}) = s(a_q, C_{t'_{j+1}}) = m^* $ holds.
Similarly, $a'_P$ and $a'_q$ interact at $\Ch{j+1}{z+1} \rightarrow \Ch{j+1}{z+2}$ and $\Ch{j+1}{z+2} \rightarrow C'_{t'_{j+1}}$, and, they keep their state $m^*$. 
Thus, $s(a'_i, C'_{t'_{j+1}}) = s(a_i, C_{t'_{j+1}})$ holds for any $i$ ($P-1 \ge i \ge 0$) and $s(a'_q, C'_{t'_{j+1}}) = s(a'_P, C'_{t'_{j+1}}) = s(a_q, C_{t'_{j+1}}) = m^* $ holds.

By the assumption, for any $j$, the number of $f(m^*)$-agents is one more than the number of $\overline{f(m^*)}$-agents in $C_{t'_j}$.
Thus, for any $j$, the number of $f(m^*)$-agents is two more than the number of $\overline{f(m^*)}$-agents in $C'_{t'_j}$.
Hence, $E'$ never converges to a stable configuration.
During the execution segment $C'_{t'_{j-1}}$, $e'_j$, when $j$ is even, agents in $A'-\{a'_{P}\}$ interact each other.
Similarly,  when $j$ is odd, agents in $A'-\{a'_q\}$ interact each other.
Moreover, for $j > 0$, at $\Ch{j}{z+1} \rightarrow \Ch{j}{z+2}$ and $\Ch{j}{z+2} \rightarrow C'_{t'_{j}}$, $a'_q$ and $a'_P$ interact.
Thus, although $E'$ is weakly-fair, $E'$ cannot solve the problem.
This is a contradiction.
Therefore, the lemma holds.
\end{proof}

By Lemma \ref{pro:sink}, there exists either a sink pair or a sink state.
However, in the both cases, $Alg_{sym}$ does not work.
Therefore, we can obtain the theorem.

\setcounter{theorem}{2}

\section{Proof of Algorithm \ref{alg1}}
\label{app:alg}

In the following, we prove the correctness of Algorithm \ref{alg1}.
\begin{theorem}
\label{the:kpart}
Algorithm \ref{alg1} solves the uniform $k$-partition problem.
This means that, in the model with an initialized BS, there exists an asymmetric protocol with $P$ states and arbitrary initial states that solves the uniform $k$-partition problem under weak fairness, where $P$ is the known upper bound of the number of agents.
\end{theorem}

To prove the theorem, we show the following two lemmas.

\begin{lemma}
\label{lem:Mp}
Let $E=C_0, C_1, \ldots$ be a weakly-fair execution of Algorithm \ref{alg1}. In configuration $C_i$ ($i \ge 0$), for any $s$ with $0 \le s \le M-1$, at least one agent with state $s$ exists.
\end{lemma}
\begin{proof}
We show the proof by induction on the index of a configuration in execution $E$.

The base case is vacuously true because $M$ is initialized to $0$ in the initial configuration $C_0$.

For the induction step, assume that the lemma holds in $C_k(0 \le k)$.
That is, in $C_k$, at least one agent with state $s$ exists for any $s$ with $0 \le s \leq M-1$.
We consider two cases of interaction at transition $C_k \rightarrow C_{k+1}$.

First, consider the case that the BS joins the interaction.
If the BS interacts with an agent with state less than $M$, the BS and the agent do not change their states.
If the BS interacts with an agent with state $M$ or more, the BS assigns $M$ to the agent and then increases $M$.
Thus, the lemma holds at $C_{k+1}$.

Next, consider the case that the BS does not join the interaction.
When two non-BS agents interact, at least one agent keeps its state.
Since $M$ is not changed, the lemma holds at $C_{k+1}$.
\end{proof}

\begin{lemma}
\label{lem:Mn}
Let $E = C_0, C_1, \ldots$ be a weakly-fair execution of Algorithm \ref{alg1}. There exists a configuration $C_i$ such that $M=n$ holds.
\end{lemma}
\begin{proof}
For contradiction, assume that $M$ does not become $n$ in $E$. Since $M$ is monotonically increasing, the BS eventually stops updating $M$. Let $l < n$ be the last value of $M$, and $C_j$ be the first configuration with $M=l$.

First, we can say that all agents have states less than $l$ after $C_j$. Otherwise some agent $a'$ with state $l'$ ($l' \ge l$) exists after $C_j$. From the algorithm, $a'$ never decreases its state unless it interacts with the BS. Since $E$ is weakly-fair, $a'$ eventually interacts with the BS. At that time, $a'$ has state at least $l' \ge l = M$ and consequently the BS increases $M$. This contradicts the assumption.

From Lemma \ref{lem:Mp}, for every $s$ with $0 \le s \le l-1$, at least one agent with state $s$ exists at configuration $C_j$. In addition, since each agent has one of states $0$ to $l-1$ with $l<n$, at least two agents have state $q$ for some $q < l$. From the algorithm, an agent with state less than $l$ ($=M$) changes its state only when it interacts with the agent with the same state. Hence, two agents with state $q$ eventually interact, and then one of them enters state $q+1$. If $q+1 \le l-1$ holds, at least two agents with state $q+1$ exist and similarly one of them enters state $q+2$. Hence eventually some agent enters state $l$. This is a contradiction.
\end{proof}

From Lemma \ref{lem:Mn}, the BS eventually sets $M=n$, and after that it never updates $M$. From Lemma \ref{lem:Mp}, when $M=n$ holds, for every $s$ with $0 \le s \le n-1$, exactly one agent has state $s$.
Clearly, the configuration achieves uniform $k$-partition and no agent updates its state after that.
Therefore, Theorem \ref{the:kpart} holds.

\end{document}